\documentclass [11pt]{article}
\usepackage{amsmath,amsthm,amsfonts,amscd,eucal,latexsym,amssymb}
\usepackage{epsfig,xcolor,hyperref}  

\input xy
\xyoption{all}
\def\sp{\hskip -5pt} 
\def\spa{\hskip -3pt} 

\def\K{{\ca K}}
\def\l{{\lambda}}
\newsymbol\bt 1202           %%% \boxtimes 

\def\cB{{\ca B}}
\def\cI{{\ca I}}
\def\cF{{\ca F}}

\def\cH{{\ca H}}

\def\cS{{\ca S}}
\def\cA{{\ca A}}
\def\cF{{\ca F}}

\def\cC{{\ca C}}

\def\cW{{\ca W}}
\def\cK{{\ca K}}

\def\sA{{\mathsf A}}

\def\sO{{\mathsf O}}

\def\sR{{\mathsf R}}

\def\sV{{\mathsf V}}

\def\bC{{\mathbb C}}           %%%  complex numbers and so on 

\def\bM{{\mathbb M}} 
\def\NN{{\mathbb N}} 
\def\bR{{\mathbb R}}

\def\bS{{\mathbb S}} 
\def\RR{{\mathbb R}}

\newsymbol \rest 1316         %%% restriction symbol 
 
\def\beq{\begin{eqnarray}}
\def\eeq{\end{eqnarray}}

\newcommand{\ca}[1]{{\cal #1}}         %%  calligraphic

\def\bx{{\boldsymbol x}}
\def\by{{\boldsymbol y}}
\def\bn{{\boldsymbol n}}

\def\supp{{\rm supp\,}}

\def\bSf{{\mathbb S}^2}  %2-sphere
\def\scri{\Im^+}         %null future boundary
\def\scrip{\Im^-}         %null past boundary
\def\tg{\widetilde{g}}
\def\tM{\widetilde{M}}
\def\tphi{\widetilde{\phi}}

\def\bms{G_{BMS}}

\def\Span{{\mathrm{Span}\,}}

\newtheorem{theorem}{Theorem}[section]

\newtheorem{corollary}[theorem]{Corollary}
\newtheorem{lemma}[theorem]{Lemma}
\newtheorem{definition}[theorem]{Definition}

\newtheorem{proposition}[theorem]{Proposition}
\theoremstyle{definition}

\newtheorem{remark}[theorem]{Remark}

\begin{document}

%%%%%%%%%%%%%   Title %%%%%%%%%%%%%%%%%%%%%%%%%% 
 
\par 

\title{The modular Hamiltonian in asymptotically flat spacetime conformal to Minkowski} 
 
%%%%%%%%%%%%%%%%%%%%%%%%%%%%%%%%%%%%%%%%%%%%% 
 
%%%%%%%%%%%% Authors %%%%%%%%%%%%%%%%%%%%%%%%%%% 

\author{ Claudio Dappiaggi\thanks{CD:
		Dipartimento di Fisica,
		Universit\`a degli Studi di Pavia \& INFN, Sezione di Pavia, \& Indam, Sezione di Pavia
		Via Bassi 6,
		I-27100 Pavia,
		Italia;
		\mbox{claudio.dappiaggi@unipv.it}} \and
	Vincenzo Morinelli\thanks{VM:
		Dipartimento di Matematica, Universit\`a di Roma Tor Vergata, Via della Ricerca Scientifica, 1, 00133 Roma (RM), Italy;
		\mbox{morinell@mat.uniroma2.it}}	
	 \and Gerardo Morsella\thanks{GM:
	 	Dipartimento di Matematica, Universit\`a di Roma Tor Vergata, Via della Ricerca Scientifica, 1, 00133 Roma (RM), Italy;
	 	\mbox{morsella@mat.uniroma2.it}}
	 	\and Alessio Ranallo\thanks{AR:
	 		Section de math\'ematiques, Universit\'e de Gen\`eve,
	 		7-9, rue du Conseil-G\'en\'eral
	 		CH-1211 Gen\`eve 4,
	 		Switzerland;
\mbox{alessio.ranallo@unige.ch}}}
 
\maketitle

\small 
\noindent {\bf Abstract}. {We consider a four-dimensional globally hyperbolic spacetime $(M,g)$ conformal to Minkowski spacetime, together with a massless, conformally coupled scalar field. Using a bulk-to-boundary correspondence, one can establish the existence of an injective $*$-homomorphism $\Upsilon_M$ between $\mathcal{W}(M)$, the Weyl algebra of observables on $M$ and a counterpart which is defined intrinsically on future null infinity $\scri\simeq\mathbb{R}\times\mathbb{S}^2$, a component of the conformal boundary of $(M,g)$. Using invariance under the asymptotic symmetry group of $\scri$, we can individuate thereon a distinguished two-point correlation function whose pull-back to $M$ via $\Upsilon_M$ identifies a quasi-free Hadamard state for the bulk algebra of observables. In this setting, if we consider $\mathsf{V}^+_x$, a future light cone stemming from $x\in M$ as well as $\mathcal{W}(\mathsf{V}^+_x)=\mathcal{W}(M)|_{\mathsf{V}^+_x}$, its counterpart at the boundary is the Weyl subalgebra generated by suitable functions localized in $\mathsf{K}_x$, a positive half strip on $\Im^+$. To each such cone, we associate a standard subspace of the boundary one-particle Hilbert space, which coincides with the one associated naturally to $\mathsf{K}_x$. We extend such correspondence replacing $\mathsf{K}_x$ and $\mathsf{V}^+_x$ with deformed counterparts, denoted by $\mathsf{S}_C$ and $\mathsf{V}_C$. In addition, since the one particle Hilbert space at the boundary decomposes as a direct integral on the sphere of $U(1)$-currents defined on the real line, we prove that also the generator of the modular group associated to the standard subspace of $\mathsf{V}_C$ decomposes as a suitable direct integral. This result allows us to study the relative entropy between coherent states of the algebras associated to the deformed cones $\mathsf{V}_C$ establishing the quantum null energy condition.}

\bigskip


\normalsize
 \bigskip

\tableofcontents

\section{Introduction}
Algebraic Quantum Field Theory (AQFT) aims at describing quantum and relativistic models with an infinite number of degrees of freedom using the language and methods of von Neumann algebras. This provides in particular several techniques for the mathematically rigorous analysis of physical quantities such as entropy and energy inequalities in QFT. 

These have become the subject of increasing interest in the last decades, stemming mainly from the fact that relations involving entropic quantities, such as the Bekenstein bound, the generalized second law of thermodynamics, and several energy inequalities, play a key role both in QFT in presence of black holes and in the quest to develop a quantum theory of gravity (see, e.g.,~\cite{Cas08, Wal12, FLPW16} and references therein). Yet, a precise mathematical formulation of these concepts turned out to require operator algebraic methods. 

Among various entropic quantities in quantum theory, relative entropy plays a fundamental role. It measures the ``additional'' average information to be gained from the outcomes of an experiment performed on a system assumed to be in a state $\psi$, but which is actually in a different state $\phi$. For quantum systems which can be described by a type I von Neumann algebra it can be expressed in terms of the density matrices $\rho_\phi$, $\rho_\psi$ as
\[
S(\phi \| \psi)=\operatorname{Tr}(\rho_\phi (\ln\rho_\phi-\ln\rho_\psi)).
\]

However, typically the local von Neumann algebras of QFT are type III and normal states can not be represented as density matrices, which entails that the above formula becomes meaningless. Notwithstanding, a more general definition in terms of the relative modular operator of the two states is possible, as shown by Araki \cite{Ar}. 

Moreover, for coherent states of free field theories, relative entropy is a  first quantization quantity. 
Indeed, it is possible to define the modular group of a given suitable real subspace of the one particle space (a standard subspace). The associated second quantization turns out to be the  modular group of the underlying von Neumann algebra with respect to the vacuum~\cite{LRT78}. This leads to the concept of entropy of a vector with respect to a standard subspace, which directly connects with the relative entropy between the corresponding coherent state and the vacuum~\cite{CLR20}.

At the same time,  the Bisognano-Wichmann property claims that the modular groups 
of the algebras of suitable regions 
have a geometrical interpretation, namely they are implemented by suitable one parameter groups of symmetries of the underlying spacetime. The prominent example is the implementation of the modular group of the Rindler wedge by Lorentz boosts. Thus, the interplay between the Tomita--Takesaki modular theory and the Bisognano--Wichmann property allowed in recent years the study and the computation of several entropic quantities in QFT \cite{CTT17, Wit18, LX18, Llocst, CLR20, MTW,  CLRR22, LM23, LM24, LPM25}. 

Most of these examples refer to models where there is a rich structure of  geometrical symmetries of the spacetime, that allows to use the Bisognano-Wichmann property. This fails in a general curved spacetime, where the lack of symmetries prevents a natural definition of a unique vacuum state. Yet the algebraic approach to quantum field theory has proven to be the natural framework to bypass this hurdle, see \cite{Brunetti:2015} for a review. More precisely, whenever one considers a free field theory on a globally hyperbolic spacetime, it is possible to associate to it unambiguously a C$^*$-algebra of observables which encompasses structural properties, such as dynamics, causality and canonical commutation/anti-commutation relations. The standard probabilistic interpretation proper of quantum theories is recovered subsequently from the GNS theorem once an algebraic state has been chosen. As mentioned above, the lack of isometries in a general scenario entails that one cannot hope to individuate a counterpart of the Poincar\'e vacuum on every background. However, one can single out a distinguished class of states whose elements are characterized by the property of abiding by the Hadamard condition, which is a constraint on the singular structure of the underlying two-point correlation function, see \cite{Radzikowski:1996pa,Radzikowski:1996ei}. 

This feature entails several far-reaching properties among which noteworthy are the existence of a local and covariant notion of Wick polynomials and the guarantee that the quantum fluctuations of all observables are finite. Despite the structural advantages of considering Hadamard states, for several years a weak point has been the lack of explicit examples unless one considers highly symmetric spacetimes. Yet many interesting scenarios are not falling in this category, a notable example being the Schwarzschild solution of the vacuum Einstein's equations which describes the exterior, static region of a spherically symmetric black hole. This represents one specific instance of a vast class of backgrounds, which are characterized heuristically speaking by the property of tending along all lightlike directions to the Minkowski metric at infinity. These spacetimes are called asymptotically flat and, when $(M,g)$ is also globally hyperbolic, it can be embedded as an open submanifold of a larger, globally hyperbolic spacetime $(\widetilde{M},\widetilde{g})$ such that $\widetilde g$ is conformally related to $g$ on $M$, see {\it e.g.} \cite{Friedrich, Wald}. The (conformal) boundary of $M$ into $\widetilde M$ includes a distinguished component, $\scri$, the {\em future null infinity} of $M$. Such asymptotic structure enjoys several notable properties which will be outlined in the main body of this work. Among these, we highlight that $\Im^+$ is universal, namely its geometric structure is exactly the same on every asymptotically flat spacetime. 
On the one hand, this implies the existence of a distinguished subgroup of the diffeomorphism group of $\Im^+$, which preserves the relevant, geometric data: the Bondi-Metzner-Sachs (BMS) group. On the other hand, if we consider a massless, conformally coupled real scalar field on any asymptotically flat spacetime $(M,g)$ and its associated Weyl algebra of observables $\mathcal{W}(M)$, we can use the underlying causal propagation to find an injective $*$-homomorphism from $\mathcal{W}(M)$ into $\mathcal{W}(\Im^+)$. This is a Weyl algebra built out of a symplectic vector space of functions living intrinsically on $\Im^+$.

In \cite{Dappiaggi:2005ci, Dappiaggi:2009fx, Dappiaggi:2017kka} the BMS group has been exploited to single out on $\Im^+$ a BMS-invariant bi-distribution with a prescribed singular structure. On the one hand, this identifies a quasi-free state for $\mathcal{W}(\Im^+)$. On the other hand, by means of the injective map from the bulk to the boundary algebra, we can induce a unique, quasi-free, Hadamard state on $\mathcal{W}(M)$ which is, in turn, invariant under the action of all background isometries. Consequently both a reference state and then the associated net of von  Neumann algebras in the bulk are obtained from boundary data. {Similar ideas have also appeared in~\cite{Ho00}.}

On account of the universal structure both of future null infinity and of the field theoretic data thereon, it is natural to try to combine modular theory techniques and this holographic approach to study entropy quantities and energy inequalities for QFTs on any asymptotically flat spacetime starting from an analysis on $\scri$. As a first step in this project, in the present paper we focus on spacetimes which are conformal to Minkowksi since, on this class of backgrounds, we can obtain explicit formulae and test the effectiveness of our approach. A prominent and physically relevant example in this family is the cosmological de Sitter spacetime. More precisely we start with a globally hyperbolic spacetime $(M,g) =(\RR^4,\chi^2\eta)$, with $\chi$ a strictly positive function on $M$, while $\eta$ is the Minkowksi metric. 
Due to the structure of the BMS-invariant state on $\scri$, we describe the boundary QFT one particle space as a direct integral on the 2-sphere of  $U(1)$-currents. Subsequently, making use of the Huygens' principle, we can identify the bulk and the asymptotic theory of one particle states in forward light cones of $M$. On account of the Bisognano-Wichmann property of the $U(1)$-current we can give an explicit first quantization formula for the modular Hamiltonian of deformed forward light-cones, a counterpart of Equation (3.26) in \cite{CTT17}. Moreover, we compute the relative entropy for the von Neumann algebra associated to these regions between coherent states. This allows us to check the quantum null energy condition (QNEC) and  the strong super-additivity of the relative entropy in analogy with \cite{MTW}. Finally, we give an expression for the relative entropy in terms of the classical traceless stress-energy tensor associated to the corresponding solution of the wave equation on $M$.

{The paper is structured as follows: First of all we focus on recalling those geometric and analytic definitions that play a key role in the algebraic approach to relative entropy and energy inequalities. In Section 2 we recall the standard subspace structure, the $U(1)$-current model as well as the entropy of a wave, a first quantization quantity strongly related to the relative entropy between coherent states. Subsequently we introduce the geometric data at the heart of our work. In Section 3 we discuss the quantum fields on $\scri$ and  on $M$ which we are interested in comparing. In Section 4 we analyse the structure of the net of real subspaces of the one particle space on $\scri$, and discuss its relation with the corresponding net of the one particle space on $M$ for suitable regions, via the holographic map. This is the used to compute the one-particle modular Hamiltonian and the relative entropy quantities for such regions. Lastly, in an outlook, we discuss how these techniques could be extended to general asymptotic flat spacetimes, highlighting difficulties and perspectives.

Huzihiro Araki made fundamental contributions to the theory of real subspaces of complex Hilbert spaces and of the associated second quantization von Neumann algebras \cite{Ara63}, as well as to the concept of relative entropy between normal states on arbitrary von Neumann algebras \cite{Ar}. Both ideas have been highly influential and underpin the central themes of this work.

}

\section{Preliminaries}

In this section we give a succinct survey of the key structures at the heart of this paper.

\subsection{Standard subspaces}\label{Sec: Standard Subspaces}
{We recall here the main results of the modular theory of standard subspaces of a complex Hilbert space, cf.~\cite{LRT78} and \cite[Ch.\ 2]{Lo08} for further details and proofs.} 

Let $\cH$ be a complex Hilbert space {with scalar product denoted by $\langle\cdot,\cdot\rangle$}, whereas $H\subset\cH$ is a closed real subspace. We say that $H$ is \textit{cyclic} if $H+iH$ is dense in $\cH$, \textit{separating} if $H\cap iH=\{0\}$ and \textit{standard} if it is cyclic and separating. 
Real subspaces come endowed with a symplectic complement 
$$H'=\{h\in\cH:\Im\langle h,k\rangle=0,\forall k\in H \},$$
and $H\subset \cH$ is called \textit{factorial} if $H\cap H'=\{0\}$.
Given a standard subspace $H\subset\cH$ the Tomita operator associated to $H$ is defined as
$$\cS_H:H+iH\ni h+ik\longmapsto h-ik\in H+iH.$$ 
This is a closed antilinear involution whose polar decomposition reads 
$$\cS_H=J_H\Delta^{1/2}_H.$$
Here $J_H$ is an anti-unitary operator, called modular conjugation, while $\Delta_H$, called modular operator, is {positive} and self-adjoint on $\cH$. In addition they abide by the following properties:
$$J_H H=H',\qquad \Delta^{it}_H H=H,\qquad J_H\Delta_H J_H=\Delta_H^{-1}.$$
{The operator $\log \Delta_H$ is called the \textit{modular Hamiltonian} associated to $H$.}

We say that a pair of standard subspaces $K\subset H$ is a {\em half-sided modular inclusion} if
\begin{align}\label{Eq: Half Sided Modular Inclusion}
\Delta_{H}^{-it}K\subset K \ \text{ for } t\geq 0.
\end{align}
These are noteworthy since they {entail the existence of} a positive energy representation of translations and dilations constructed as follows. We denote by $\mathbf{P}$ the group of affine transformations of $\RR$, that is dilations act as
$\delta(s) x=e^{s}x$,
while translations as
$\tau(s) x = x + s$, $x{, t, s} \in \RR$. In this framework the point $x=1$ is fixed with respect to the combined action
$\delta_1(s)x=e^{s}(x-1)+1$. {A unitary, strongly continuous representation of $\mathbf{P}$ on a complex Hilbert space is said to be of positive energy if the self-adjoint generator of the translations is a positive operator. One has the following key result~\cite[Thm.\ 2.4.1]{Lo08}.}
\begin{theorem}{\label{theorem-1-standard}
Given $K\subset H$, a half-sided modular inclusion of standard subspaces of a complex Hilbert space $\cH$, there exists a positive energy representation {U} of $\mathbf{P}$ {on $\cH$ such that}
\[
 U(\delta(2\pi t))=\Delta_H^{-it},\qquad U(\delta_1(2\pi t))=\Delta_K^{-it}{,\qquad t \in \RR}.
\]
Furthermore, the generator $P$ of the translation group is $\frac{1}{2\pi}\left(\log(\Delta_K)- \log(\Delta_H) \right)$ and it holds that $\log(\Delta_{U(\tau(s))H}) = \operatorname{Ad} U(\tau(s))(\log(\Delta_H)) = \log(\Delta_H) + 2\pi sP$.
}
\end{theorem}

\subsection{Entropy of a vector} 
Given a factorial standard  subspace $H$ of $\cH$, the \emph{cutting projection} associated to $H$ is
\[
P_H(\phi+\phi')= \phi \ , \qquad \phi\in H, \ \phi'\in H' \, .
\]
It turns out {\cite{CLR20}} that $P_H$ is a densely defined, closed, real linear operator satisfying  
\[
P^2_H = P_H \  \ , \ \ P_H^* =P_{iH}  = -iP_H i\ \ , \ \  P_H\Delta^{is}_H=
\Delta^{is}_H P_H\;\forall s\in\bR \ .
\]
Whenever $H$ is a factorial standard subspace, the \emph{entropy} of a vector $\phi\in\cH$ with respect to $H$ is 
\begin{equation}\label{Sf}
S_H(\phi)= -\Im\langle\phi, P_H i\log\Delta_H \phi\rangle{,}
\end{equation}
{or, more precisely, $S_H$ is the quadratic form associated to the real self-adjoint operator $P_Hi \log\Delta_H$, while $S_H(\phi) := +\infty$ if $\phi$ is not in the form domain of this operator.} If $H$ does not fall in the class considered here, refer to \cite{CLRR22}, although, in this paper, we shall not {need} these cases. Given, $\phi, \phi_n, \psi\in\cH$ as well as $H,H_n$, {$n \in \NN$, }closed, real linear subspaces of $\cH$, the entropy enjoys the following properties \cite{CLR20, LM23}:
\begin{itemize}
\item \textit{Positivity:} $S_H(\phi)\geq 0$ or $S_H(\phi) = +\infty$;
\item \textit{Monotonicity:} If $K \subset H$, then $S_K(\phi)\leq S_H(\phi)$;
\item \textit{Lower semicontinuity:} If $\phi_n  \to \phi$, then $S_H(\phi) \leq \liminf_n S_H(\phi_n)$;
\item \textit{Monotone continuity: }If $H_n\subset H$ is an increasing sequence with $\overline{\bigcup_n H_n} = H$, then $S_{H_n}(\phi) \to S_H(\phi)$;
\item \textit{Locality:} $S_H(\phi+\psi) = S_H(\phi)$ if $\psi\in H'$, moreover $S_H(\phi)= 0$ if and only if ${\phi}\in H'$;
\item \textit{Unitary invariance:} $S_{UH}(U \phi) = S_H(\phi)$ for any unitary operator $U$ on $\cH$.
\end{itemize}

{The entropy of a vector }with respect to a standard subspace is directly connected to Araki's relative entropy between coherent states on the corresponding second quantization von Neumann algebra, as we now briefly explain. We refer the reader to App.~\ref{app:secondquant} for our notations and conventions on second quantization.

Given a von Neumann algebra $\cA$ on $\cH$ and two cyclic and separating vectors $\xi,\eta\in\cH$, they identify two normal and faithful states {$\omega$} and {$\varphi$} on $\cA$. In turn one can construct the relative modular data out of $\cS_{{\xi,\eta}}$, that is the closure of the anti-linear map {$\cA\eta\ni a\eta\mapsto a^*\xi\in\cA\xi$}. Its polar decomposition {$\cS_{\xi,\eta}=J_{\xi,\eta}\Delta_{\xi,\eta}^{1/2}$} identifies the relative modular conjugation $J_{{\xi,\eta}}$ and the relative modular operator $\Delta_{{\xi,\eta}}$. In view of \cite{Ar}, Araki's relative entropy reads
$${S_\cA(\varphi \| \omega):=-\langle\eta,\log\Delta_{\xi,\eta} \eta\rangle.}$$

Let us consider {the second quantization} von Neumann algebra $R(H) \subset \mathcal B(\cF(\cH))$ as per Equation~\eqref{Eq: VN algebras} {in App.\ \ref{app:secondquant}}, associated to $H$, a standard subspace of $\cH$. We call {\em coherent} a normal state thereon of the form $\omega_\psi:=\omega\circ\operatorname{Ad}W(\psi)$,
where $\psi\in\cH$ while $\omega$ is the Fock vacuum. As shown in \cite[Prop.\ 4.2]{CLR20}, it holds that
\begin{equation}\label{eq:Ent2}
 S_{R(H)}(\omega_\psi \|\omega)=S_H(\psi).
\end{equation}
Observe that Equation~\eqref{eq:Ent2} has a much wider range of applicability than one could expect at first glance since, for all $\phi,\psi\in \cH$, see \cite[Sec. 3]{Llocst},

\begin{equation}\label{eq:cost}
S_{R(H)}(\omega_\phi \| \omega_\psi)=S_{R(H)}(\omega_{\phi-\psi} \| \omega).
\end{equation} 

\subsection{The $U(1)$-current}\label{sect:U1}
{We denote by $\mathcal{S}(\RR)\equiv\mathcal{S}(\RR;\RR)$ the space of real valued Schwartz functions on $\RR$, and by $\mathcal{S}'(\RR)$ its topological dual, the real valued tempered distributions.} Let us consider the following locally convex, topological vector space 
\begin{equation}\label{Eq: X Space}
{X:=\{f\in\cS(\RR)\;|\;\hat f\in  L^2(\RR_+,E\,dE)\},}
\end{equation}
where $\RR_+:=(0,\infty)$ {and $\hat f(E) := \frac1{\sqrt{2\pi}}\int_\RR f(u) e^{iEu}\, du$ is the Fourier transform}. {We consider the quotient space  $(X+\RR)/\RR$, which can be turned into a real Hilbert space, $\cH^{(1)}$ upon completion with respect to the inner product}
\begin{equation}\label{Eq: Real Scalar Product}
\langle f,h\rangle_\bR= {\frac12}\int_{\RR} {|E|}\hat f(-E)\hat h(E)\,dE.
\end{equation}
Observe that, with a slight abuse of notation, we denote with $f$ an equivalence class of functions $[f]\in\cH^{(1)}$. The space
$\cH^{(1)}$ is equipped with a complex structure induced by $\iota: X\to X$ which is completely identified by its action on $\mathcal{S}(\bR)$, namely $\mathcal{S}(\bR)\ni f\mapsto\widehat{\iota f} = i \,{\rm sign}(E)\widehat f(E)$.
In addition we consider a symplectic form $\beta$ on $X$ which is also defined starting from its action on $\mathcal{S}(\bR)$, namely
\begin{equation}\label{symp}
{\mathcal{S}(\bR)\times}\mathcal{S}(\bR)\ni (f,h)\mapsto\beta(f,h)=\frac{1}{2}\int_\RR f'(u)h(u)du.
\end{equation}

\noindent We exploit Equations \eqref{Eq: Real Scalar Product} and \eqref{symp} to identify the following complex scalar product on $\cH^{(1)}$:
\[
\langle f,h\rangle=\langle f,h\rangle_\bR+i \beta(f,h) ={ \int_0^{+\infty} \hat f(-E) \hat h(E) E\,dE = \lim\limits_{\epsilon\to 0^+}\frac 1 {2\pi} \int_{\RR\times\RR} \frac{f(u) h(u')}{(u-u'-i\epsilon)^2}du\,du' \,.}
\]
The space $\cH^{(1)}$ carries the irreducible, positive energy, unitary representation $U^{(1)}$ of $\mathrm{PSL}(2,\RR)$ with lowest weight $1$, whose action {is defined on real smooth functions on the one-point compactification of $\RR$ as}
\begin{equation}\label{Eq: U representation of PSL(2,R)}
(U^{(1)}(g)f)(u)=f(g^{-1}u)\, ,
\end{equation}
{and extended by density to $\cH^{(1)}$. In the above formula, the action on $\RR$ of }$g=\left(\begin{array}{cc} a&b\\c&d
\end{array}\right)\in\mathrm{SL}(2,\RR)$ {is defined by }$gu=\frac{au+b}{cu+d}$. In particular, we will make use of the 1-parameter subgroups of $\mathrm{SL}(2,\RR)$ of translations and dilations, defined respectively by
\begin{equation}\label{eq:transdil}
\tau_{\RR_+}(s)u := u+s, \qquad \delta_{\RR_+}(s)u := e^s u, \qquad u, s \in \RR.
\end{equation}

Denoting by $\cI_0$ the set of open intervals of the real line, the M\"obius covariant net $H^{(1)}$ of standard subspaces of $\cH^{(1)}$ associated with {the $U(1)$-current} is given by
\[
\cI_0\ni I\mapsto H^{(1)}(I)=\big\{[f] \in \cH^{(1)}:   f \in \cC^\infty_0(\RR), \,\mathrm{supp}\,f\subset I\big\}^-\,. 
\]
Here we have reintroduced equivalence classes to avoid a source of confusion.
\begin{remark}\label{rem:propertiescurrent}
Let $\cI$ be the set of {non-dense} open intervals of the compactified line $\overline{\RR}\equiv\bS^1$ and {for $I \in \cI$ let $I'$ be the interior of} $\bS^1\setminus I$.  We observe that the net
$$\cI_0\ni I\mapsto H^{(1)}(I){\subset \cH^{(1)}},$$
extends to a counterpart on $\bar \RR$ and it satisfies the following properties:
 \begin{enumerate}
\item \emph{Isotony:} $I_1 \subset  I_2 \implies H^{(1)}(I_1) \subset H^{(1)}(I_2)$ ; 
\item \emph{Locality:} $I_1 \subset I'_2 \implies H^{(1)}(I_1) \subset H^{(1)}(I_2)'$;
\item \emph{Möbius covariance:} There exists a unitary, positive energy representation $U^{(1)}$ of the M\"obius group on $\cH{^{(1)}}$ such that $U^{(1)}(g)H^{(1)}(I) = H^{(1)}(gI)$, $I\in\cI$;
\item \emph{Standardness:} $H^{(1)}(I) $ is standard for every $I\in\cI$;
\item \emph{Haag duality:} $H^{(1)}(I') = H^{(1)}(I)'$ for all $I\in \cI$; 
\item \emph{Bisognano-Wichmann property:} $U{^{(1)}}\big(\delta_I(-2\pi s)\big)=\Delta_{H^{(1)}(I)}^{is}$ where $\delta_I$ is the one-parameter group of ``dilations'' acting on $I${, i.e., $\delta_I(s) = g_I^{-1} \delta_{\RR_+}(s) g_I$, with $g_I \in \mathrm{SL}(2,\RR)$ such that $g_I I = \RR_+$}.
\end{enumerate}
It follows from the Bisognano-Wichmann property that, if $I_1\subset I_2$ is an inclusion of intervals on the circle with the same (anti-clockwise) endpoint,  then $H^{(1)}(I_1)\subset H^{(1)}(I_2)$ is an half-sided modular inclusion, see Equation \eqref{Eq: Half Sided Modular Inclusion}. Most notably, this result also applies in the real-line picture when $I_1\subset I_2\subset\RR$ where both $I_1$ and $I_2$ are half-lines.
\end{remark}

{\noindent For future reference, we highlight the following standard subnets of $H^{(1)}$: 
$$\cI_0 \ni I \mapsto H^{(k)}(I):=\{[f^{(k-1)}]\in\cH^{(1)}: f\in\cC_0^\infty(\RR),\supp f\subset I\}^-\subseteq H^{(1)}(I), \qquad 1\leq k\in\NN,$$
called the one particle net of the $(k-1)$-th derivative of the $U(1)$ current.  We recall that, if $I\subset\RR$ is an unbounded interval,
\begin{equation}\label{eq:k1}
H^{(k)}(I)=H^{(1)}(I),
\end{equation}
see \cite[Sect.~2.5]{GLW98}
\medskip
}

{\em Entropy of a vector for a $U(1)$-current.} In the framework above, we can construct a net of von Neumann algebras on $\bS^1$ associated to a $U(1)$-current. As in the previous discussion we can consider the relative entropy between a coherent state and the vacuum. This has been first computed in \cite{Llocst}. More precisely, consider $I_t=(t,+\infty)\subset\RR$, $H^{(1)}(I_t)$ its standard subspace and $\cA(I_t)=R(H^{(1)}(I_t))$ the associated von Neumann algebra. Given $h\in\cS(\bR){\subset \cH^{(1)}}$ 
\begin{equation}\label{eq:U1ent}
  S_{\cA(I_t)}(\omega_{h}\|\omega)=S_{H^{(1)}(I_t)}(h)=\pi\int_t^{+\infty}(u-t) (h'(u))^2du.
\end{equation} 
Using a density argument, in \cite{CLRR22}, Equation \eqref{eq:U1ent} has been extended to every $h\in\cH$. An analogous formula on bounded intervals has been derived in \cite{LM24}.

\subsection{Spacetimes conformal to Minkowski}\label{sec:geometry} 

For a smooth manifold $M$, we will denote by $\cC^k(M)$, $k =0,1,2,\dots,\infty$, the space of real valued $k$ times continuously differentiable functions on $M$, and by $\cC_0^k(M)$ the subspace of the compactly supported ones. As customary, unless confusion can arise, we will always identify points of $M$ with their representation in any one of its coordinate system to simplify the notation. Consequently, we will denote by the same symbol a function in different coordinate systems.

In this paper we shall consider four-dimensional, globally hyperbolic manifolds which are \textit{conformal to Minkowski spacetime}, namely $(M,g)\equiv(\mathbb{R}^4,\chi^2\eta)$ where $\eta$ is the flat metric, here considered with signature $(-,+,+,+)$ while $\chi\in \cC^\infty(\mathbb{R}^4,(0,\infty))$. This is a specific example of an \textit{asymptotically flat spacetime} \cite[Chap. 11]{Wald}, namely there exists a smooth background $(\widetilde{M},\widetilde{g})$ such that $M$ turns out to be an open submanifold of $\widetilde{M}$ with boundary $\Im \subset \widetilde{M}$. Without loss of generality we shall only consider those cases for which there exists a subset $\scri\subset\Im$, called {\em future null infinity}. 
This is the maximal component of $\Im$ satisfying $\scri \cap J^-_{\tM}(M) = \emptyset$\footnote{Given a globally hyperbolic spacetime $(M,g)$ with an open subset $\sA\subset M$, we denote by $I^\pm_{M}(\sA) = \bigcup_{x\in \sA}I^\pm_{M}(x)$ and by $J^\pm_{M}(\sA)=\bigcup_{x\in \sA}J^\pm_{M}(x)$, where $I^\pm_{M}(x)$ and $J^\pm_{M}(x)$ are respectively the chronological and the causal future ($+$) and past ($-$) of $x\in M$, see \cite[Chap.\ 8]{Wald}.} and it is an embedded submanifold of $\tM$ diffeomorphic to $\mathbb{R}\times\mathbb{S}^2$. For the sake of completeness we outline succinctly the construction of $\widetilde{M}$ { for $(M,g)$} in App. \ref{app:a}. 

Furthermore, as one can infer from Appendix \ref{app:a}, in our scenario, $(\tM,\tg)$ is globally hyperbolic and 
$\widetilde{g}|_M= \Omega^2|_M {g}$ where $\Omega \in \cC^\infty(\tM)$ and $\Omega>0$ on $M$. On $\scri$, $\Omega =0$ and $d\Omega \neq 0$. 
Moreover, defining $n^a := \tg^{ab} \partial_b \Omega$, 
there exists a smooth function, $\omega$, defined in $\tM$ with $\omega >0$ on $M\cup \scri$, such that 
$\widetilde{\nabla}_a (\omega^4 n^a)=0$ on $\Im$ and the integral lines of $\omega^{-1} n$ are complete on $\scri$. Observe that it is possible to introduce the notion of {\em past} null infinity $\scrip$ defined, mutatis mutandis, as $\scri$ \cite[Chap. 11]{Wald}.

{
	\begin{remark}\label{rem:range}
		Observe that, although we consider spacetimes $(M,g)$ such that $g=\chi^2\eta$ with $\chi\in \cC^\infty(\mathbb{R}^4,(0,\infty))$, most of our analysis and results require only the existence of future null infinity. Therefore we could enlarge the class of admissible manifolds so that $(M,g)$ is a globally hyperbolic spacetime conformal only to a subset of $(\bR^4,\eta)$ and admitting $\Im^+$ as part of its conformal boundary. Yet, since this class of background is rather implicit in its definition, we prefer to focus on a restricted, but more explicit collection of spacetimes.
	\end{remark}
}

\begin{remark}\label{rem:embedded}
Henceforth we shall always consider {$M$} realized as {an open subset} in $\widetilde{M}$. Consequently all causal structures will be defined with reference to $\widetilde{M}$ and for simplicity of the  notation we shall write $I^\pm(x)$ and $J^\pm(x)$ in place of $I^\pm_{\widetilde{M}}(x)$ and of $J^\pm_{\widetilde{M}}(x)$ for all $x\in\widetilde{M}$. {We will instead keep the subscript and write $I^{\pm}_M(x), J^{\pm}_M(x)$, $x \in M$, for the corresponding subsets of $M$. Moreover, given a subset $\sA\subset \tM$, we denote by $\sA^\prime$ the causal complement of $\sA$ in $\tM$, namely the interior of the set of all points of $\tM$ that cannot be connected to any point of $A$ by a causal curve. One can then define $\sA^{\prime\prime}$, the causal completion of $\sA$ in $\tM$,  by taking twice the causal complement, $\sA^{\prime\prime}:=(\sA^\prime)^\prime$.}
\end{remark}

Considering any asymptotically flat spacetime, therefore those conformal to Minkowski in particular,  the metric structures of $\scri$ are affected by a {\em gauge freedom} due to the possibility of changing the metric $\tg$ in a neighborhood of $\scri$
with a factor $\omega$ smooth and strictly positive. It corresponds to the map $\Omega \to \omega \Omega$ which does not affect the differentiable structure of $\Im^+$. More precisely, fixing $\Omega$, $\scri$  turns out to be
the union of all future-oriented  integral lines of the field 
$n^a :=\widetilde{g}^{ab}\widetilde{\nabla}_b\Omega$.
This property is, in fact, invariant under gauge transformation, but the field $n$
depends on the gauge. 
For a fixed background $(M,g)$ in the class that we consider,
the manifold $\scri$ together with its degenerate metric $\widetilde{h}$ induced by $\tg$ and
 the field $n$ on $\scri$
form a triple $(\scri,\widetilde{h},n)$ which, under gauge transformations $\Omega \to \omega \Omega$, transforms as
\beq
\scri \to \scri \:,\:\:\:\:\: \widetilde{h} \to \omega^2 \widetilde{h} \:,\:\:\:\:\: n \to \omega^{-1} n \label{gauge}\:.
\eeq
We denote by $C$ the equivalence class of all triples under the action of this map. Observe that $C$ is {\em universal} for all asymptotically flat spacetimes within those considered in this work. More precisely, if $C_1$ and $C_2$ are two classes of triples associated respectively to $(M_1,g_1)$
and $(M_2,g_2)$ there exists a diffeomorphism $\gamma: \scri_1 \to \scri_2$ such that for suitable representatives $(\scri_1,\widetilde{h}_1, n_1)\in C_1$
and $(\scri_2,\widetilde{h}_2, n_2)\in C_2$, 
\beq
\gamma(\scri_1) = \scri_2 \:,\:\:\:\:\: \gamma^* \widetilde{h}_1=\widetilde{h}_2 \:,\:\:\:\:\:\gamma^* n_1=n_2 \nonumber\:.
\eeq
The proof of this statement  relies on the following nontrivial result \cite{Wald}. 
For any asymptotically flat spacetime $(M,g)$ (either $(M_1,g_1)$ or $(M_2,g_2)$ in particular) and any initial 
choice for $\Omega_0$,
varying the latter with a judicious choice of $\omega$, 
one can always fix $\Omega := \omega \Omega_0$ in order that the metric $\tg$ associated with $\Omega$ satisfies
\beq
\tg\spa\rest_{\scri}  = -2du \:d\Omega +  d\Sigma_{\bS^2}(\theta,\varphi)\:. \label{met}
\eeq
This formula uses the fact that in a neighborhood of $\scri$, $(u, \Omega,
\theta,\varphi)$ define a meaningful coordinate system.
Here $ d\Sigma_{\bS^2}(\theta,\varphi)$ is the standard metric on a unit $2$-sphere and $u \in \bR$ is nothing but an affine parameter along
the {\em complete} null geodesics forming $\scri$ itself with tangent vector $n= \partial/\partial u$. In these coordinates $\scri$ is just the set of the points with
$u \in \bR$, {$\Omega = 0$,} $(\theta,\varphi) \in \bS^2$ and, no-matter the initial spacetime $(M,g)$ 
(either $(M_1,g_1)$ or $(M_2,g_2)$ in particular), one has the triple 
$(\scri,\widetilde{h}_B, n_B) := (\bR\times \bS^2, d\Sigma_{\bSf}, \partial/\partial u)$. {Such a coordinate system $(u,\theta,\varphi)$ on $\scri$ is called a \textit{Bondi frame}.}

\begin{remark}\label{Rem: Trick Limit}
With reference to the compactification procedure outlined in Appendix \ref{app:a}, observe that the r\^{o}le of $u$ in Equation \eqref{met} is played by the null coordinate $u=t-r$, while {$\Omega=\chi^{-1}\cos(\arctan(v))$} where $v=t+r$. Consequently, evaluating a quantity at $\scri$ can be done equivalently either on $M$ taking the limit $v\to\infty$ or on $\widetilde{M}$ taking the limit $\Omega\to 0^+$. In the following we will implicitly exploit this feature.
\end{remark}

\begin{remark}\label{Rem: BMS}
\noindent Future null infinity is particularly noteworthy for the existence of a distinguished subgroup of $\mathrm{Diff}(\Im^+)$. This is the {\bf Bondi-Metzner-Sachs (BMS) group}, $G_{BMS}$ \cite{Penrose, Penrose2, Geroch, AS}, which consists of all
$\gamma\in\mathrm{Diff}(\scri)$ which preserve the  universal structure of $\scri$, i.e. $(\gamma(\scri),\gamma^*\widetilde{h}, \gamma^*n)$
differs from $(\scri,\widetilde{h},n)$ at most by a gauge transformation as in Equation \eqref{gauge}. It turns out that $G_{BMS}$ can be viewed as the semidirect product between $SO(3,1)^{\uparrow}$ and $\cC^\infty(\bS^2)$ seen as an Abelian group under addition. The elements of this subgroup are called {\bf supertranslations}. We do not delve into a detailed analysis of $G_{BMS}$ and of its structural properties referring an interested reader to \cite{Dappiaggi:2005ci}.
\end{remark}

\section{Quantum fields on $\scri$ and on $M$}

In this section we work with $(M,g)\equiv(\mathbb{R}^4,\chi^2\eta)$ with signature $(-,+,+,+)$ as per Section \ref{sec:geometry} with $\chi\in \cC^\infty(M;(0,\infty))$, see Remark \ref{rem:range} for a possible generalisation. We endow $M$ with the standard Cartesian coordinates $(t,\bx)$, $\bx\equiv (x_1,x_2,x_3)$. 
We will also use on $M$ and $\tM$ the coordinate systems $(t,r,\theta,\varphi)$, $(u,v,\theta, \varphi)$ and $(U,V,\theta,\varphi)$ introduced in App.~\ref{app:a}.

\subsection{Scalar field on $\scri$}\label{QFT1} 
We consider a quantum field theory on $\scri$ following a construction first outlined in \cite{Dappiaggi:2005ci}. 
We construct a field theory on $\scri$ based on smooth scalar fields $\psi$ and assuming $G_{BMS}$ as the natural symmetry group, see Remark \ref{Rem: BMS}. This should be read, more appropriately, as a QFT on the class
of all triples $(\scri,\widetilde{h}, n)$ connected with $(\scri, \widetilde{h}_B, n_B)$ by the transformations of $\bms$. In the following we report only the key ingredients and results, which will play a significant r\^{o}le in our discussion. A more detailed exposition can be found in \cite{Dappiaggi:2017kka}.

Our starting point is 
\begin{equation}\label{Eq: Functions on Scri}
	\cS(\scri)=\{\psi\in \cC^\infty(\Im^+)\;|\;\lim\limits_{|u|\to\infty}|u|^k\partial^\alpha_u\psi=0,\;\forall\alpha,k\in\mathbb{N}\cup\{0\}\},
\end{equation}
where we consider a fixed Bondi frame $(u,\theta,\varphi)$ and the limits hold uniformly in the angular coordinates
 $\cS(\scri)$. It can be proven, see \cite{Dappiaggi:2005ci}, that $\cS(\scri)$ does not depend on the choice of Bondi frame and it can be equipped with a symplectic form invariant under the action of the BMS group.

\begin{theorem}\label{theo1}
Given $\cS(\scri)$ as per Equation \eqref{Eq: Functions on Scri}, the map $\sigma: \cS(\scri)\times \cS(\scri) \to \bR$
\beq
\sigma(\psi_1,\psi_2) := {\frac12}\int_{\bR\times \bS^2} 
\left(\psi_2 \frac{\partial\psi_1}{\partial u}  - 
\psi_1 \frac{\partial\psi_2}{\partial u}\right) 
du \wedge \epsilon_{\bS^2}(\theta,\varphi)\:, \label{sigma}
\eeq
$\epsilon_{\bS^2}$ being the standard volume form of the unit $2$-sphere, 
 is a nondegenerate symplectic form on $\cS(\scri)$ independently from the chosen Bondi frame $(u,{\theta, \varphi})$.
\end{theorem}

\begin{definition}\label{Def: C*-alg_bulk} 
We call \emph{algebra of observables on $\scri$}, $\mathcal{W}(\scri)$, the unique (up to $*$-isomorphisms) C$^*$-algebra associated to $(\cS(\scri),\sigma)$, whose generators $\cS(\scri)\ni \psi\mapsto W(\psi)$ abide by 
\begin{equation}\label{Eq: Weyl Generators}
W(\psi)^*=W(-\psi),\qquad W(\psi)W(\psi^\prime)=e^{\frac{i}{2}\mathcal{\sigma}(\psi,\psi^\prime)}W(\psi+\psi^\prime),\quad\forall \psi,\psi^\prime\in \cS(\scri).
\end{equation}
\end{definition}

Following \cite{Dappiaggi:2005ci}, we can associate to $\mathcal{W}(\Im^+)$ a unique, quasi-free, BMS invariant algebraic state
$$\omega_{\Im^+}:\mathcal{W}(\Im^+)\to\mathbb{C},$$
which is completely determined by its action on the Weyl generators, {\it i.e.}, for all $\psi,\psi^\prime\in\mathcal{S}(\Im^+)$
\begin{gather}
\omega_{\Im^+}(W(\psi))=e^{-\frac{\omega_2(\psi,\psi)}{4}},\notag\\
\omega_2(\psi,\psi^\prime)=\lim\limits_{\epsilon\to 0^+}{\frac1 {2\pi}}\int\limits_{\mathbb{R}^2\times\mathbb{S}^2}dS^2(\theta,\varphi)du \,du^\prime\,\frac{\bar{\psi}(u,\theta,\varphi)\psi^\prime(u^\prime,\theta,\varphi)}{(u-u^\prime-i\epsilon)^2}, \label{eq:2point}
\end{gather}
where $(u,\theta,\varphi)$ is a Bondi frame on $\scri$. It is worth noticing that, as proven in \cite{Dappiaggi:2005ci}, $\omega_{\Im^+}$ is the unique BMS-invariant ground state on $\mathcal{W}(\Im^+)$.

Focusing on the pair $(\mathcal{W}(\Im^+),\omega_{\Im^+})$, it turns out, see \cite[Thm.\ 2.2]{Dappiaggi:2005ci}, that the associated GNS triple is $(\mathcal{F}(\mathcal{H}),\Pi,\Omega_{\Im^+})$, where $\mathcal{F}(\mathcal{H})$ is the Bosonic Fock space built out of the one particle Hilbert space $\mathcal{H}$ {defined as the completion of $\cS(\scri)$ with respect to the scalar product defined by the two-point function $\omega_2$ in~\eqref{eq:2point}. We note that we have the unitary equivalences}
\begin{equation}\label{eq:hilbsp}
\cH \simeq L^2(\RR_+\times\mathbb{S}^2,E\,dE d\mathbb{S}^2)\simeq\int^{{\oplus}}_{\mathbb{S}^2}L^2(\RR_+,E\,dE) d\mathbb{S}^2,
\end{equation}
{where the first} equivalence is given by the unitary map
\begin{eqnarray}
\nonumber
{\cH} &\longrightarrow &L^2(\RR_+\times\mathbb{S}^2,E\,dEd\mathbb{S}^2)\\
\nonumber
f(u,\theta,\varphi)& \longmapsto &\hat f(E,\theta,\varphi),
\end{eqnarray}
where $$\hat f(E,\theta,\varphi)=\frac1{\sqrt{2\pi}}\int_\RR f(u,\theta,\varphi)e^{iuE}du$$
is the one-dimensional Fourier transform of $f(u,\theta,\varphi)$ in the $E$-variable restricted to
 $\RR_+$. The second unitary equivalence is well known, see Equation \eqref{integral=tensor-prod} {in App.~\ref{app:directintegral}}.
 The representation $\Pi:\mathcal{W}(\Im^+)\to\mathcal{B}(\mathcal{F}(\mathcal{H}))$ is completely determined by its action on the Weyl generators, namely 
$$\Pi(W(\psi))=e^{i\Psi(\psi)},\quad\forall\psi\in\mathcal{S}(\Im^+),$$
where $\Psi(\psi)=ia(\psi_+)-ia^\dagger(\psi_+)$, $a,a^\dagger$ being the creation and annihilation operators acting on $\mathcal{F}(\mathcal{H})$, while
{}$$\psi_+(u,\theta,\varphi):=\int\limits_0^\infty\frac{dE}{\sqrt{2\pi}} e^{-iE u}\widehat{\psi}(E,\theta,\varphi),$$
For any connected, open subset $\mathcal{O}\subset\Im^+$, we can consider the pair $(\mathcal{W}(\mathcal{O}),\left.\omega_{\Im^+}\right|_{\mathcal{O}})$, where $\mathcal{W}(\mathcal{O})$ is the C$^*$-subalgebra of $\mathcal{W}(\Im^+)$ generated by all $\psi\in\mathcal{S}(\Im^+)$ such that $\textrm{supp}(\psi)\subset\mathcal{O}$. 

For the sake of notation simplicity, in what follows, we will refer to one of the previous pictures for the Hilbert space by stressing the configuration $(u,\theta,\varphi)$ or momentum variables $(E,\theta,\varphi)$ of the function considered. 

\subsection{Conformally coupled scalar field on $M$}
On top of $M$ taken as in Section \ref{sec:geometry}, we consider a real, massless and conformally coupled scalar field $\Phi:M\to\mathbb{R}$ whose dynamics is ruled by the equation

\begin{equation}\label{eq:wave}
P\Phi\doteq\left(\Box_g-\frac R6\right)\Phi=0,
\end{equation}
where $\Box_g$ is the D'Alembert wave operator built out of the metric $g$, while $R$ is the associated scalar curvature. Since the operator $P$ is normally hyperbolic it admits unique advanced and retarded fundamental solutions \cite{BGP}, namely there exists continuous maps $G^\pm:\cC_0^\infty(M)\to\cC^\infty(M)$ such that 
$$P\circ G^\pm=\mathbb{I}\quad\textrm{and}\quad G^\pm\circ P=\mathbb{I}|_{\cC_0^\infty(M)},$$
and $\textrm{supp}(G^\pm(f))\subseteq J^\mp(\textrm{supp}(f))$ for all $f\in\cC_0^\infty(M)$. On the one hand, as a byproduct of the Schwartz kernel theorem, we can associate to $G^\pm$ unique bi-distributions $\mathcal{G}^\pm\in\mathcal{D}^\prime(M\times M)$, while, on the other hand, we can define the {\em causal (or Pauli-Jordan) propagator} $\mathcal{G}=\mathcal{G}^--\mathcal{G}^+$. The latter entails that, for any $f\in\cC_0^\infty(M)$, there exists an associated solution of Equation \eqref{eq:wave}
$$\Phi_f\doteq\mathcal{G}(f)\in\cC^\infty(M),$$
where the right-hand side is the partial evaluation of $\mathcal{G}$. All these data can be recollected in a convenient setting -- see \cite[Prop 3.3]{Benini:2013fia} for a proof in a more general scenario.

\begin{proposition}\label{prop:sympl_sol}Let $(M,g)$ be as in Section \ref{sec:geometry} and let $P$ be as per Equation \eqref{eq:wave}. Denoting by $\mathcal{G}$ the associated causal propagator, then the pair $(X,\mathcal{G})$ where $X\doteq\frac{\cC_0^\infty(M)}{P[\cC^\infty_0(M)]}$ identifies a symplectic vector space.
\end{proposition}	

{We denote by $\mathcal{W}(M)$ the unique (up to $*$-isomorphism) C$^*$-algebra associated to the symplectic space $(X,\mathcal{G})$ introduced in Theorem \ref{theo1}, whose generators $W([f])$, $[f] \in X$, abide by the same relations as in Equation \eqref{Eq: Weyl Generators} with $\sigma$ replaced by $\mathcal{G}$  and $\cS(\scri)$ by $X$, i.e.
\begin{equation}\label{Eq: Weyl Generators}
W([f])^*=W([-f]),\qquad W([f])W([f^\prime])=e^{\frac{i}{2}\mathcal{G}(f,f^\prime)}W([f+f^\prime]),\quad\forall [f],[f^\prime]\in X.
\end{equation}
}

 Observe that the vector space $X$ is in one-to-one correspondence with the set of smooth, spacelike compact solutions of Equation \eqref{eq:wave}, namely those $\Phi\in\cC^\infty(M)$ such that $P\Phi=0$ and $\textrm{supp}(\Phi)\cap\Sigma$ is compact, where $\Sigma$ is any Cauchy surface of $M$. We can associate to $(X,\mathcal{G})$ an algebra of observables following \cite{BR2, BGP}.

In order to define an inclusion map from the algebra of observables in the bulk to the counterpart on $\Im^+$, the first step consists of rewriting Equation \eqref{eq:wave} on $(\tM,\tg)$, introduced in Section \ref{sec:geometry}. This is a standard construction, see in particular \cite[App. D]{Wald}, which we recall here succinctly. To start with, observe that, denoting by $\Omega$ the conformal factor introduced still in Section \ref{sec:geometry}, to any $\Phi_f\in \cC^\infty(M)$, solution of Equation \eqref{eq:wave} generated by $f\in \cC^\infty_0(M)$, we can associate 
\begin{equation}\label{eq:conf_wave}
\Psi\doteq\Omega^{-1}\Phi_f\in \cC^\infty(M)\quad\textrm{such that}\quad\left(\Box_{\tg}-\frac{\widetilde{R}}{6}\right)\Psi=0,
\end{equation}
where $\Box_{\tg}$ and $\widetilde{R}$ are the D'Alembert wave operator and the scalar curvature built out of $\tg$. Equation \eqref{eq:conf_wave} stands on the following proposition, proven in \cite{Dappiaggi:2005ci}.

\begin{proposition}\label{prop2}\em Assume that $(M,g)$ and $(\tM,\tg)$ are as per Section \ref{sec:geometry} where $\tg\sp\rest_M = \Omega^2 g$ for a given $\Omega\in\cC^\infty(\widetilde{M})$ with $\Omega>0$. Consider any but fixed open set $\widetilde{\sV}\subset \tM$ with 
	$\overline{M\cap J^-(\scri)}\subset \widetilde{\sV}$ (the closure being referred to $\tM$) such that both $(\widetilde{\sV}, \tg)$ and $(M\cap \sV,g)$ are globally hyperbolic. If $\phi : M\cap\widetilde{\sV}  \to \bC$ 
	has compactly supported Cauchy data on some Cauchy surface of $M\cap \widetilde{\sV}$ and if it satisfies the
	massless conformal Klein-Gordon equation,
	\beq
	\Box \phi - \frac{1}{6} R \phi =0\label{cc}\:,
	\eeq
	{\bf(a)} the field $\tphi := \Omega^{-1}\phi$ can be  extended uniquely to a smooth solution 
	in $(\widetilde{\sV}, \tg)$
	of 
	\beq
	\widetilde{\Box} \tphi - \frac{1}{6} \widetilde{R} \tphi =0\label{cc2}\:;
	\eeq  
	{\bf(b)} for every smooth positive factor $\omega$ defined in a neighborhood of $\scri$ used to rescale 
	$\Omega \to \omega \Omega$ in such a neighborhood,
	$(\omega\Omega)^{-1} \phi$
	extends uniquely to a smooth field $\psi$ on $\scri$.
\end{proposition}

In this work, we always consider $\tM$ to be globally hyperbolic and therefore it plays the r\^{o}le of $\widetilde{V}$ in Proposition \ref{prop2}. On account of the above data, we can infer that, being $\Psi$ spacelike compact, there must exist $\widetilde{f}\in \cC^\infty_0(M)$ such that $\Psi=\widetilde{\mathcal{G}}(\widetilde{f})|_M$, where $\widetilde{\mathcal{G}}$ is the causal propagator of $\Box_{\tg}-\frac{\widetilde{R}}{6}$. Using \cite[Lemma 2.2]{Pinamonti:2009zqj}, it turns out that the propagators $\mathcal{G}$ and $\widetilde{\mathcal{G}}$ are related by the following expression
\begin{equation}\label{Eq:forseserveforseno}
	\left.\widetilde{\mathcal{G}}\right|_{\cC^\infty_0(M)}=\Omega^{-1}\circ\mathcal{G}\circ\Omega^3:\cC^\infty_0(M)\to\cC^\infty(M).
\end{equation}
In view of this result we have established a $1:1$ correspondence between solutions of Equation \eqref{eq:wave} with smooth and compactly supported initial data on $M$ and those of Equation \eqref{eq:conf_wave} with initial data lying in the same space. 

Since $\Im^+$ is a smooth, codimension 1 submanifold of $\tM$, we can also define the restriction thereon of an element of $\cC^\infty(\tM)$. In turn, putting together all the data gathered in the discussion above, we can define the following map
\begin{equation}\label{Eq: Upsilon}
\Upsilon_{M}: X\to \cC^\infty(\Im^+),\quad [f]\mapsto \left.\widetilde{\mathcal{G}}(\Omega^{-3}f)\right|_{\Im^+}.
\end{equation}
{Using this map, an algebra correspondence  stems from Huygens' principle, which holds true only since we are considering massless scalar fields on backgrounds which are conformal to the four-dimensional Minkowski spacetime \cite{Friedlander}. We can state the following proposition, whose proof can be found in \cite[Thm 4.5]{Dappiaggi:2005ci}.}

\begin{proposition}\label{Prop:b2b_correspondence}Let $(X,\mathcal{G})$ be the symplectic space as per Proposition \ref{prop:sympl_sol} and let $(\mathcal{S}(\Im^+),\sigma)$ be the symplectic space on $\Im^+$ as per Theorem \ref{theo1}. Then the map $\Upsilon_{M}$ is an injective symplectomorphism. As a consequence and in view of Definition \ref{Def: C*-alg_bulk}, there exists an injective $*$-homomorphism $\Gamma_{M}:\mathcal{W}(M)\to\mathcal{W}(\Im^+)$ which is completely characterized by its action on the generators, namely, for all $[f]\in X$,
	$$\Gamma_{M}(W[f])=W(\Upsilon_{M}([f])).$$
\end{proposition}

{The following lemma is a consequence of the previous proposition.}

\begin{lemma}\label{Lem:ci_crediamo_tanto}

Let $\mathcal{W}({I_M^+(x)})\subset\mathcal{W}(M)$ be the Weyl C$^*$-algebra generated by all $f\in\cC^\infty_0(M)$ such that $\textrm{supp}(f)\subset {I_M^+(x)}$. It holds that $\Gamma_{M}(\mathcal{W}({I_M^+(x)})\subset\mathcal{W}({ I^+(x)\cap\Im^+})$ where the right hand side is the Weyl $*$-subalgebra of $\mathcal{W}(\Im^+)$ generated by all $\psi\in\mathcal{S}(\Im^+)$ such that $\textrm{supp}(\psi)\subseteq {I^+(x)\cap\Im^+}$.
\end{lemma}

 In addition, as a consequence of Proposition \ref{Prop:b2b_correspondence}, the BMS-invariant state $\omega_{\Im^+}$ on $\cW(\scri)$ identifies via pull-back a state on the Weyl algebra $\mathcal{W}(M)$ as follows:
\begin{gather}
	\omega_M\doteq\Gamma^*_M\omega_{\Im^+}:\mathcal{W}(M)\to\bC\notag\\
	\omega_M(W([f]))\doteq\omega_{\Im^+}(\Gamma_{M}(W([f]))),\quad\forall [f]\in X. \label{eq:omegaM}
\end{gather}
It turns out that $\omega_{M}$ coincides with the Poincar\'e vacuum if $M$ is Minkowski spacetime. In addition, for any admissible bulk background $M$, $\omega_M$ is always a quasi-free, Hadamard state, invariant under the action of all isometries \cite{Moretti:2006ks}. 
We can restrict it to ${I^+(x)}$ and, in view of {Lemma \ref{Lem:ci_crediamo_tanto},}
$$\left.\omega_{M}\right|_{\mathcal{W}({I_M^+(x)})}=\Gamma^*(\left.\omega_{\Im^+}\right|_{I^+(x)\cap\Im^+}).$$

We can then use the state $\omega_M$ to define a net of von Neumann algebras indexed by open regions $\sO \subset M$:
\begin{equation}\label{eq:AO}
\cA(\sO) := \{\pi_M(W([f])\,:\, f  \in \cC^\infty_0(\sO)\}'' \subset B(\cH_M),
\end{equation}
where $(\pi_M, \cH_M,\Omega_M)$ is the GNS representation  to $\omega_M$.

\section{Modular Hamiltonian and QNEC for deformed cones}\label{sect:4}
\subsection{Regions in $M$ and $\Im^+$}
For future convenience, first of all we introduce some notable regions. As in the previous section, we shall always work with a Bondi frame $(u,\theta,\varphi)$ on $\scri$ with $u\in\bR$.
{Moreover, we will employ the standard parametrization $\bn = \bn(\theta,\varphi) := (\sin \theta \cos \varphi, \sin\theta \sin \varphi,\cos\theta)$ of $\mathbb{S}^2$.}

\begin{definition}\label{def:regions}
Given a point $x \in M$, we set the following notation
\begin{equation}\label{Eq: Kx}
\mathsf{V}_x^+:=I^+(x)\cap M\quad\textrm{and}\quad \mathsf{K}_x :=  I^+(x)\cap\Im^+.
\end{equation} 
Given $C\in \cC^0(\mathbb{S}^2)$ a positive \textit{half-strip region} in $\Im$  is a locus of the form
\begin{equation}\label{Eq: SC}
\mathsf{S}_{C}:=\{(u,\theta,\varphi)\in\scri: C(\theta,\varphi)<u\}.
\end{equation}
\end{definition}

\noindent In the following lemma we relate these two regions.

\begin{lemma}\label{lem:KIm}
Let $x{=(t,\bx)}\in M$, and let $C_x\in \cC^\infty(\mathbb{S}^2)$ be {defined by} $C_x(\bn) := t - \bx\cdot\bn$, $\bn \in \mathbb S^2$. Then $$\mathsf{K}_x=\left\{(u,\theta,\varphi)\;:\; C_x(\theta,\varphi)<u \right\}\subset \Im^+\,.$$ In particular, $\mathsf{K}_x =\mathsf{S}_{C_x}$ is a positive half-strip region.
\end{lemma}
\begin{proof}
{
Given $y =(u,\bn) \in \mathsf{K}_x$, there is a timelike curve $\gamma : [0,1] \to M \cup \scri$ such that $\gamma(0) = x$, $\gamma(1) = y$ and $\gamma(s) =(u_s,v_s,\bn_s) \in I_M^+(x)$ for all $s \in [0,1)$. Since the causal structure is invariant under conformal rescaling, it follows that the Minkowskian square
\[
(\gamma(s) -x)^2 = -u_s v_s-u_x v_x+u_x v_s+u_s v_x -\frac12(v_s-u_s)(v_x-u_x)(\bn_s \cdot \bn_x-1) < 0,
\]
where $x = (u_x,v_x,\bn_x)$. This entails, eventually for $s \to 1^-$ (since in this limit $v_s \to +\infty$),
\[
u_s > \frac{v_s[u_x-\frac12(v_x-u_x)(\bn_s\cdot \bn_x-1)]-u_x v_x}{v_s-v_x -\frac12(v_x-u_x)(\bn_s\cdot\bn_x-1)},
\]
which for $s \to 1^-$ provides $u \geq u_x-\frac12(v_x-u_x)(\bn\cdot \bn_x-1)= C_x(\bn)$. At the same time, on account of Lemma \ref{lem:limitcomponents}, the point $(C_x(\bn), \bn) \in \scri$ is lightlike to $x$, so we conclude that $u > C_x(\bn)$, and then $\mathsf{K}_x \subset \left\{(u,\theta,\varphi)\;:\; C_x(\theta,\varphi)<u \right\}$.

To prove the converse inclusion, consider $y = (u, \bn) = (T, R,\bn) \in \scri$ with $u > C_x(\bn)$, and fix $u_z \in \RR, v_z > 0$ such that $u > u_z > C_x(\bn)$ and
\[
u_z > \frac{v_z C_x(\bn)-u_x v_x}{v_z-u_x-v_x+C_x(\bn)},
\]
which is possible since the right hand side converges to $C_x(\bn)$ as $v_z \to +\infty$. As seen above, this implies $z := (u_z,v_z,\bn) = (T_z, R_z, \bn) \in I_+(x) \cap M$. Moreover, $u > u_z$ is equivalent to $T-T_z > R-R_z$, and $T_z+R_z = V_z < \frac\pi 2 = T+R$, so that $|R-R_z| < T-T_z$. As a consequence, the curve $\gamma(s) := (T_z+s(T-T_z), R_z+s(R-R_z), \bn)$, $s \in [0,1]$, connects $z$ to $y$ and is such that $\|\dot\gamma(s)\|^2 = -(T-T_z)^2+(R-R_z)^2 < 0$ thanks to Equation~\eqref{eq:einsteinstatic}, so that $y \in \mathsf{K}_x$.
}

\end{proof}

\begin{remark}
{With the notations of Lemma \ref{lem:KIm}, observe} that, given $x,y\in M$, the sets $\mathsf{K}_x, \mathsf{K}_y \subset \Im^+$ as per Equation \eqref{Eq: Kx} are such that, if $y\in J^+(x)$, then $\mathsf{K}_y \subset \mathsf{K}_x$. Furthermore $\mathsf{K}_x=\mathsf{K}_y$ if and only if $x=y$. Indeed, from $t-\bx\cdot\bn=t_y-\by\cdot\bn$ for all $\bn\in \mathbb{S}^2$, we can infer that $(\bx-\by)\cdot\bn=t-t_y=\mathrm{const.}$ for all $\bn\in\mathbb{S}^2$. Since $\bn$ is arbitrary this entails that the constant must be vanishing which implies in turn that $\bx=\by$ and therefore $t=t_y$.
\end{remark}

\subsection{One particle net on $\Im^+$}
Recalling that on $\Im^+$ we are considering the Hilbert space $\mathcal{H}$ identified in Equation \eqref{eq:hilbsp}, we can define the real one particle subspace associated to an open region $\sO\subset\Im^+$ as
\begin{equation}\label{def:H}
\sO\longmapsto H(\sO)=\overline{\{\psi\in\cC_0^\infty(\Im^+)\;\textrm{such that}\;\textrm{supp} (\psi)\subset \sO\}}^{\|\cdot\|_\cH}\,.
\end{equation}
In the following we prove a result on the structure of the one-particle Hilbert space on $\scri$ in a setting which is slightly more general that the one needed in this work.

\begin{lemma}\label{lem:dirint} Consider $c_1,c_2\in \mathbb{R}\cup\{\pm\infty\}$ with $c_1< c_2$ as well as {an open set} $A\subset\bS^2$. If we denote by
$$\sR=\{(u,\theta,\varphi)\;|\; c_1<u<c_2\;\textrm{and}\;(\theta,\varphi)\in A\}\subset\scri,$$
then the local subspace $H(R)$ decomposes as
     \begin{align}\label{eq:disrect}
       H(\sR) = \int^{\oplus_\RR}_{A} H^{(1)}(c_1,c_2)\, d\bS^2\subseteq \int^{\oplus_\RR}_{\bS^2}\cH^{(1)} d\bS^2.
          \end{align}
 Then $H(\sR)$ is separating for any $\sR = (c_1,c_2)\times A$ where $(c_1,c_2)\subset\mathbb R$ is proper,
  and standard if {also $\bS^2 \setminus A$ is a null set}. 
\end{lemma}
\begin{proof}
On account of Equation~\eqref{eq:hilbsp}, the  Hilbert space $\cH$ can be identified with $\int_{\mathbb{S}^2}^{\oplus} \cH^{(1)}d\mathbb{S}^2$. Furthermore, setting $I := (c_1,c_2)$, there holds the following unitary equivalence: 
    $$H^{(1)}(I)\otimes_\RR L^2(A;d\mathbb{S}^2)\simeq\int^{{\oplus_\RR}}_{A}  H^{(1)}(I)d\mathbb{S}^2 \subset\cH.$$
However, notice that { the elements $f \otimes h$ with $f \in \cC^\infty_0(I)$, $h \in \cC_0^\infty(A)$, form a total set in $ H^{(1)}(I)\otimes_\RR L^2(A,d\mathbb{S}^2)$, and that} $H(\sR)$ contains all functions of the form {$(u,\theta,\varphi) \mapsto f(u)h(\theta,\varphi)$ for all such $f$ and $h$. It is then sufficient to show that these functions form a total set in $H(\sR)$ too. To this end, observe first that if $\psi \in H(\sR) \subset L^2(\RR_+\times \mathbb{S}^2, E\,dEd\mathbb{S}^2)$, there exists a sequence $\{\psi_k\} \subset \cC^\infty_0(\scri)$ with $\supp \psi_k \subset \sR$ such that
\[
\|\psi-\psi_k\|^2 = \int_{\bS^2} d\bS^2(\bn) \int_0^{+\infty} dE\,E |\hat\psi(E,\bn)-\hat\psi_k(E,\bn)|^2 \to 0.
\]
This entails, passing to a subsequence if necessary, that $\psi(\cdot,\bn)$ is the limit in $\cH^{(1)}$ of $\{\psi_k(\cdot,\bn)\} \subset \cC^\infty_0(I)$ for almost all $\bn \in \bS^2$. Therefore $\psi(\cdot,\bn) \in H^{(1)}(I)$ for almost all $\bn \in \bS^2$, and $\psi(\cdot,\bn) = 0$ for almost all $\bn \not\in A$.
Assume now that $\psi \in H(R)$ is such that, for all $f \in \cC^\infty_0(I)$, $h \in \cC_0^\infty(A)$,
\[
\langle \psi, f h \rangle_{\cH} = \int_{A}d\bS^2(\bn) h(\bn) \int_0^{+\infty}dE\,E \hat \psi(-E,\bn) \hat f(E) = 0.
\]
By the density of $\cC_0^\infty(A)$ in $L^2(A,\bS^2)$, it follows that for all $f \in\cC^\infty_0(I)$ there is a null set $N_f \subset \bS^2$ such that
\[
\int_0^{+\infty}dE \,E\hat\psi(-E,\bn) \hat f(E) = 0, \qquad \forall \bn \in \bS^2\setminus N_f.
\] 
In turn, since $\cC^\infty_0(I)$ is dense in $H^{(1)}(I)$ which is separable, we can find a sequence $\{f_k\} \subset \cC^\infty_0(I)$ dense in $H^{(1)}(I)$, so that the last equation together with the fact that $\psi(\cdot, \bn) \in H^{(1)}(I)$ finally implies that $\psi = 0$. Consequently the functions $(u,\theta,\varphi) \mapsto f(u)h(\theta,\varphi)$, $f \in \cC^\infty_0(I)$, $g \in \cC_0^\infty(A)$, form a total set in $H(\sR)$ and Equation~\eqref{eq:disrect} is proven.} 

The statements about the separating and cyclicity properties of $H(R)$ follow from the above direct integral decomposition and from Lemma \ref{lem:propertiesstandard}{, observing that
\[
H(\sR) = \int_{\bS^2}^{\oplus_\RR} H^{(1)}(I) \chi_A(\bn) \,d\bS^2(\bn) \quad \Rightarrow \quad H(\sR)' =  \int_{\bS^2}^{\oplus_\RR}[ H^{(1)}(I)' \chi_A(\bn)+\cH^{(1)}\chi_{\bS^2\setminus A}(\bn)] \,d\bS^2(\bn),
\]
and recalling that $H^{(1)}(I)$ is cyclic and separating for all proper intervals $I \subset \RR$, and that $H(\sR)$ is cyclic if and only if $H(\sR)'$ is separating.}
\end{proof}

Given $C\in\cC^0(\bS^2)$, we define a distorted lightlike translation as the following unitary operator on the direct integral picture of $\cH$:

   \begin{align}\label{eq:non-const-translation}
    (T_C \xi)(u,\theta,\varphi) =  \xi(u-C(\theta,\varphi), \theta,\varphi).
   \end{align}

\noindent Similarly, we define the distorted lightlike dilations as
   \begin{align}\label{eq:non-const-dilations}
    (D_C\xi)(u,\theta,\varphi) = \xi(e^{-C(\theta,\varphi)}u, \theta,\varphi).
   \end{align}
If $C$ is constant, these maps coincide with the usual translations and the dilations.

\begin{proposition}\label{pr:null-covariance}
The family of real subspaces $H(\sO)$, see Definition \ref{def:H}, indexed by open connected regions $\sO\subset\Im^+$, is covariant with respect to $T_C, D_C$, $C\in\cC^0(\bS^2)$, namely 
$$T_C H(\sO) = H(\sO+C)\quad\textrm{and}\quad D_{C} H(\sO) = H(e^{C}\cdot \sO),$$
where $\sO + C = \{(u + C(\theta,\varphi), \theta,\varphi): (u,\theta,\varphi) \in \sO\}$
and $e^C\cdot \sO = \{(e^{C(\theta,\varphi)}u, \theta,\varphi): (u, \theta,\varphi) \in \sO\}$ {are open subsets of $\scri$}.
\end{proposition}

\begin{proof}
On the Hilbert space $\cH=\int_{\mathbb{S}^2}^\oplus \cH^{(1)}\;d\mathbb{S}^2$ distorted translations and dilations act as
$$T_C= \int_{\mathbb{S}^2}^\oplus U^{(1)}(\tau_{\RR_+}(C(\theta,\varphi)))\;d\mathbb{S}^2 \;\;\text{and}\;\;D_C= \int_{\mathbb{S}^2}^\oplus U^{(1)}(\delta_{\RR_{+}}(C(\theta,\varphi)))\;d\mathbb{S}^2,$$
where $U^{(1)}$ {and $\tau_{\RR_+}$, $\delta_{\RR_+}$ are} defined in Equations\ \eqref{Eq: U representation of PSL(2,R)}, {\eqref{eq:transdil}.}
The operators $T_C$ and $D_C$ are unitaries on $\cH$ since they are defined as direct integrals of unitary operators. Considering  $\xi\in\cC_0^\infty(\Im^+)$ supported in $\sO$, in the direct integral picture we have that $(T_C\xi)(u,\theta,\varphi)=\xi(u-C(\theta,\varphi),\theta,\varphi)$ and $(D_C\xi)(u,\theta,\varphi)=\xi(e^{-C(\theta,\varphi)}u,\theta,\varphi)$  are smooth functions compactly supported in $\sO + C $ and $e^C\cdot \sO$ respectively. By a density argument {and unitarity} we can infer the covariance property.
\end{proof}
	
\noindent The net $\Im^+\supset\sO\mapsto H(\sO)$ satisfies the following properties:
\begin{enumerate}
\item \textit{Isotony:} $$\sO_1\subseteq \sO_2,\; \text{then }\;H(\sO_1)\subseteq H(\sO_2)\,;$$
\item \textit{Locality:} $$\sO_1\cap\sO_2=\emptyset\; \text{then }\;H(\sO_1)\subseteq H(\sO_2)'\,;$$
\item \textit{Standardness:} The subspaces  $H(\sO)$ are standard/cyclic/separating according to Lemma \ref{lem:dirint}; 
\item \textit{Covariance:} $$T_CH(\sO)=H(\sO+C)\,, \qquad D_C H(\sO)=H(e^C\cdot \sO)\,.$$
\end{enumerate}

\noindent As a direct consequence of the properties that $T_C$ is decomposable and $H(\sR+C)=T_C H(\sR)$, the following corollary holds true.

\begin{corollary}\label{cor:stripdis}
Consider $c_1,c_2\in \mathbb{R}\cup\{\pm\infty\}$ with $c_1< c_2$ as well as $A\subset\bS^2$ {open, and set $\sR := (c_1,c_2)\times A$}. Given $C\in\cC^0(\mathbb{S}^2)$, then 
$$H(\sR+C)=\int_A^{\oplus_\RR} H^{(1)}(c_1+C(\theta,\varphi),c_2+C(\theta,\varphi)) d\mathbb{S}^2\subseteq \int_{\mathbb{S}^2}^\oplus \cH^{(1)}d\mathbb{S}^2\,.$$
As in Lemma \ref{lem:dirint}, $H(\sR+C)$ is separating for any $\sR=(c_1,c_2)\times A$ where $(c_1,c_2)\subset \RR$ is proper, and standard if {$\mathbb{S}^2\setminus A$ is a null set}. In particular, for a continuous function $C:\bS^2\rightarrow \RR$,
\begin{gather}\label{eq:HSC}
H(\mathsf{S}_C)={T_C H((0,+\infty)\times\bS^2)=}\int_{\bS^2}^{\oplus_\RR} H^{(1)}(C(\theta,\varphi),+\infty) d\mathbb{S}^2(\theta,\varphi)
\end{gather}
is a standard subspace.
\end{corollary}

\subsection{One particle net on $M$}
{Recalling the definition of the state $\omega_M$ on $\cW(M)$ in terms of the state $\omega_{\scri}$ on $\cW(\scri)$ as per Equation ~\eqref{eq:omegaM}, by the uniqueness of the GNS construction, we obtain that the associated GNS representations $(\pi_M, \cH_M, \Omega_M)$ and $(\Pi,\cF(\cH),\Omega_{\scri})$ are related by 
\[
\cH_M \simeq \overline{\Pi(\Gamma_M(\cW(M))\Omega_{\scri}}^{\|\cdot\|_{\cF(\cH)}}\subset\cF(\cH), \qquad \pi_M \simeq \Pi\circ\Gamma_M(\cdot)|_{\cH_M}.
\]
This implies that the local von Neumann algebras of the theory on $M$ are linked to those of the theory on $\scri$ by
\[
\cA(\sO) = R(N(\sO))|_{\cH_M}, \qquad N(\sO) := \overline{\{ \Upsilon_M([f])\,:\, f \in \cC^\infty_0(\sO)\}}^{\|\cdot\|_\cH} \subset \cH,
\]
with the} map $\Upsilon_M$ defined in Equation \eqref{Eq: Upsilon}. {Thanks to its definition and to Huygens' principle, it is easy to see that $N(\sO) \subset H(\partial J^+(\sO)\cap\scri)$ for any open set $\sO \subset M$.} In addition the inclusion is proper in general, {\em cf.} Remark~\ref{rem:double} below. In particular, we {have}, for every set $\mathsf{V}_x^+$ as per Equation~\eqref{Eq: Kx}:
\begin{equation}\label{def:N}
N(\mathsf{V}_x^+)=\overline{\{\psi=\Upsilon_M([f])\in\cC_0^\infty(\Im^+)\;\textrm{with}\;\textrm{supp}(f)\subset \mathsf{V}_x^+\}}^{\|\cdot\|_\cH}\,,
\end{equation}
{and i}n the following we establish that {$N(\mathsf{V}_x^+)$ actually} recovers the whole standard space {$H(\mathsf{K}_x)$}. 

We start by proving an ancillary result.

\begin{lemma}
Given $(M,g)$ as in Section \ref{sec:geometry} and a conformal factor $\Omega = 2\chi^{-1} \cos V$ in a neighbourhood of $\Im^+$, 
where $V$ is defined as in Equation \eqref{Eq: U and V}, then, for $f \in \cC^\infty_0(M)$, it holds that
\[
\Upsilon_M([f])(u,\boldsymbol{n}) =  \int_{\RR^3} d\boldsymbol{x}\, (\chi^3 f)(u+\boldsymbol{n}\cdot \boldsymbol{x}, \boldsymbol{x}), \qquad (u,\boldsymbol{n}) \in \RR \times \mathbb{S}^2,
\]
where $\Upsilon_M$ is defined in Equation \eqref{Eq: Upsilon}.
\end{lemma}

\begin{proof}
Considering the operator $P$ as in Equation \eqref{eq:wave}, we denote by $\mathcal G_{\bM}$ its causal propagator on Minkowski spacetime, by $\mathcal G$ that on $M$ and by $\widetilde {\mathcal G}$ that on $\widetilde M$. They are connected by Equation \eqref{Eq:forseserveforseno}, namely
\[
\mathcal G = \chi^{-1} \circ \mathcal{G}_{\bM} \circ \chi^3, \qquad \widetilde{\mathcal G}|_{C_0^\infty(M)} = \Omega^{-1} \circ \mathcal G \circ \Omega^{3}.
\]
Consequently, on account of Equation \eqref{Eq: Upsilon}, $\Upsilon_M([f]) = (\Omega\chi)^{-1} \mathcal{G}_{\bM}(\chi^3 f)|_{\Im^+}$. Working in light cone coordinates $v = y_0 + r$, $u = y_0-r$, $\boldsymbol{n} = \frac1r \boldsymbol{y}$ and for $v > u$,
\[\begin{split}
\mathcal G_{\bM}(\chi^3 f)&(y_0,\boldsymbol{y}) \\
&= \frac2{v-u}\int_{\RR^3} \frac{d\boldsymbol{x}}{|\boldsymbol{n}- \frac{2\boldsymbol{x}}{v-u}|} \left[(\chi^3f)\left(\frac{v}2\left(1-\left|\boldsymbol{n}-\frac{2\boldsymbol{x}}{v-u}\right|\right)+\frac{u}2\left(1+\left|\boldsymbol{n}-\frac{2\boldsymbol{x}}{v-u}\right|\right),\boldsymbol{x}\right)\right.\\
&\phantom{\frac2{v-u}\int_{\RR^3} \frac{d\boldsymbol{x}}{|\boldsymbol{n}- \frac{2\boldsymbol{x}}{v-u}|} \left[\right.}-\left.  (\chi^3f)\left(\frac{v}2\left(1+\left|\boldsymbol{n}-\frac{2\boldsymbol{x}}{v-u}\right|\right)+\frac{u}2\left(1-\left|\boldsymbol{n}-\frac{2\boldsymbol{x}}{v-u}\right|\right),\boldsymbol{x}\right)\right].
\end{split}\]
Taking into account that, in a neighbourhood of $\Im^+$, $\Omega \chi = 2 \cos V = 2(1+v^2)^{-1/2}$, see Remark \ref{Rem: Trick Limit}, { that $|\bn-\frac{2\bx}{v-u}| =1- \frac{2\bn\cdot\bx}v+o(v^{-1})$ as $v \to +\infty$,} and that $\chi^3 f$ has compact support, it descends that the integrand of the above expression, divided by $\Omega \chi$, converges, as $v \to +\infty$, for fixed $u \in \RR$, to
\[
(\chi^3 f)(u+\boldsymbol{n}\cdot \boldsymbol{x},\boldsymbol{x}).
\]
Moreover, if we denote by $K \subset \RR^3$ the compact set which is the projection on $\RR^3$ of the support of $\chi^3 f$, we can find $c > 0$ such that the bit of the integrand between square brackets in the above integral is bounded by $c\chi_K(\bx)$, $\chi_K$ being the restriction of $\chi$ to $K$. In addition,  if $R > 0$ is such that $\bx \in K$ implies $|\bx| < R$, we can find $v$ such that
\[
\frac{2R}{v-u} < \frac12 \quad \Rightarrow \quad \left|\boldsymbol{n}-\frac{2\bx}{v-u}\right| \geq \left|1 - \frac{2|\bx|}{v-u}\right| > 1-\frac{2R}{v-u} > \frac12. 
\]
Therefore the statement follows by applying the dominated convergence theorem.
\end{proof}

\begin{proposition}\label{prop:Vx}
Given $x\in M$, the sets $\mathsf{V}_x^+ \subset M$ and $\mathsf{K}_x\subset \Im^+$ {defined in~\eqref{Eq: Kx}} satisfy 
$$N(\mathsf{V}_x^+)=H(\mathsf{K}_x)\,,$$ 
where $N(\mathsf{V}_x^+)$, respectively $H(\mathsf{K}_x)$, are defined as in Equation~\eqref{def:N}, respectively \eqref{def:H}.
\end{proposition}

\begin{proof}
As already remarked, we have that $N(\mathsf{V}_{x}^+)\subseteq H(\mathsf{K}_{x})$, since by the definition of $\Upsilon_M([f])$ in Equation \eqref{Eq: Upsilon}, {if $f \in \cC^\infty_0(\mathsf{V}_x^+)$} then $\Upsilon_M([f])\in\cC^\infty_0(\Im^+)$ and $\supp(\Upsilon_M([f]))\subseteq\mathsf{K}_x$. 
To prove the converse inclusion, we start by observing that {thanks to the previous Lemma} the $u$-Fourier transform of $\Upsilon_M([\chi^{-3}f])$ is
\[\begin{split}
\int_{\RR} \frac{du}{\sqrt{2\pi}} \,e^{i E u} \Upsilon_M([\chi^{-3} f])(u,\bn) &= \int_{\RR} \frac{du}{\sqrt{2\pi}} \,e^{iE(u+\bn \cdot \bx)} \int_{\RR^3} d\bx \,e^{-i E \bn\cdot \bx}f(u+\bn \cdot \bx,\bx) \\
&= {(2\pi)^{3/2}}\hat f(E,E\bn),
\end{split}\]
{where $\hat f$ on the right hand side denotes the 4-dimensional Fourier transform of $f$, }so that
\[
\| \Upsilon_M([\chi^{-3}f])\|^2 ={(2\pi)^3}\int_{\mathbb S^2} d\mathbb S^2(\boldsymbol{n}) \int_0^{+\infty}  dE\, E |\hat f(E,E\bn)|^2 = {(2\pi)^3}\int_{\RR^3} \frac{d\boldsymbol{p}}{|\boldsymbol{p}|} |\hat f(|\boldsymbol{p}|,\boldsymbol{p})|^2
\]
is {proportional to} the squared norm of $\hat f|_{\partial \mathsf{V}_x^+}$ in the massless Klein-Gordon field one particle space $L^2\left(\partial\mathsf{V}^+,\theta(p_0)\delta(p^2)d^4p\right)$, where we can identify $p_0=E$. As a consequence, the map associating to $f \in \cC_0^\infty(M)$ with $\operatorname{supp} f \subset \mathsf{V}_x^+$ the element $\Upsilon_M([\chi^{-3}f]) \in N(\mathsf{V}_x^+)$, extends to a unitary between the massless Klein-Gordon field standard subspace of {$L^2(\partial \mathsf{V}^+)$} associated to $\mathsf{V}_x^+$ with $N(\mathsf{V}_x^+)$. Now according to~\cite[Lemma 8.2]{BDL07}, the von Neumann algebra of the massless Klein-Gordon field associated to $\mathsf{V}_x^+$ is generated by the Weyl unitaries of the distributions
\begin{equation}\label{eq:H}
G(y) = g(y_0) P(\boldsymbol{\partial}_y) \delta(\by-\bx), \qquad y \in \RR^4
\end{equation}
with $g$ a real smooth compactly supported function with support in $(x_0, +\infty)$ and $P$ a real polynomial. This entails, by projection on the one particle space, that the standard subspace associated to $\mathsf{V}_x^+$ is the norm closure of the set $\mathcal T_x$ of distributions~\eqref{eq:H}, so that, by the above remark, the elements $\Upsilon_M([\chi^{-3}G])$, $G \in \mathcal T_x$, are dense in $N(\mathsf{V}_x^+)$. 

Moreover, we note that for $P = P_k$ a real homogeneous polynomial of degree $k$,
 $$\Upsilon_M([\chi^{-3}G])(u,\bn)=(-1)^kg^{(k)}(u+\bn\cdot\bx)P_k(\bn).$$ 
To conclude the proof, it is then sufficient to show that the above functions form a total set in $H(\mathsf{K}_x)${. To this end, recall that, thanks to Lemma~\ref{lem:KIm} and Corollary~\ref{cor:stripdis}, $H(\mathsf{K}_x) =$}$\int^{\oplus_\RR}_{\mathbb S^2} H^{(1)}(x_0 - \bn\cdot \bx, +\infty)\,d\mathbb{S}^2(\bn)$ {and} assume that $\psi \in H(\mathsf{K}_x)$ is real orthogonal to all such functions. Then
\[
\Re \langle \psi, \Upsilon_M([\chi^{-3}G])\rangle = \int_{\mathbb S^2} (-1)^k\Re\langle \psi(\bn), g^{(k)}(\cdot + \bn\cdot\bx)\rangle P_k(\bn)\,d\mathbb S^2(\bn)= 0,
\]
and since the restrictions to $\mathbb S^2$ of the real homogeneous polynomial form a total set in $L_{\mathbb R}^2(\mathbb S^2)$, this entails that $\psi(\bn) \in H^{(1)}(x_0 - \bn\cdot \bx, +\infty)$ is real orthogonal to $g^{(k)}(\cdot + \bn\cdot\bx)$ for almost every $\bn \in \mathbb S^2$. But since $H^{(1)}(x_0 - \bn\cdot \bx, +\infty)$ is separable, there is a countable set of real smooth functions $g$ with compact support in $(x_0,+\infty)$ such that  the functions $g^{(k)}(\cdot + \bn\cdot\bx)$ are dense in $H^{(k+1)}(x_0 - \bn\cdot \bx, +\infty)= H^{(1)}(x_0 - \bn\cdot \bx, +\infty)$, so that finally $\psi = 0$.
\end{proof}

{The cyclicity of $H(\mathsf{K}_x)$, the strong continuity of Weyl operators and the previous Proposition imply in particular that
\[
\cF(\cH) = \overline{R(H(\mathsf{K}_x))\Omega_{\scri}}^{\|\cdot\|_{\cF(\cH)}} = \overline{\Pi\circ\Gamma_M(\cW(\mathsf{V}_x^+) )\Omega_{\scri}}^{\|\cdot\|_{\cF(\cH)}} \subset \cH_M,
\]
i.e., actually $\cH_M = \cF(\cH)$ and
\begin{equation}\label{eq:AOR}
\cA(\sO) = R(N(\sO)) \qquad \text{for all open }\sO \subset M.
\end{equation}}

\begin{remark}\label{rem:double} A double cone in {$M$} can be defined as follows: let $x,y\in M$ with $y$ in the future of $x$, then $\mathsf{D}_{x,y}:=\mathsf{V}^+_x\cap \mathsf{V}^-_y$. According to light propagation one can define $\mathsf{K}_{x,y} := \partial J^+(\mathsf{D}_{x,y})\cap\scri$ as the causal image of $\mathsf{D}_{x,y}$ on $\Im^+$.
Similarly to the proof of Lemma \ref{lem:KIm} one can  see that if $\mathsf{D}_{x,y}$ and $\mathsf{D}_{x^\prime,y^\prime}$ are two double cones in {$M$} which are spacelike separated, then $\mathsf{K}_{x,y}\cap \mathsf{K}_{x^\prime,y^\prime}\neq\emptyset$. In particular, taking $\psi=\Upsilon_M([f])$ with $\supp \,f \subseteq\mathsf{D}_{x,y}$ and $\psi'=\Upsilon_M([f'])$ with $\supp \,f' \subseteq\mathsf{D}_{x',y'}$, we see that $\sigma(\psi,\psi')=0$ ($\sigma$ as in \eqref{sigma}) even if {in general} $\supp\,\psi\cap\supp\,\psi'\neq\emptyset$. In particular one concludes that   $N(\mathsf{D}_{x,y})$ is not dense in $H(\mathsf{K}_{x,y})$. At the same time, if $\mathsf{D}_{x,y}$ and $\mathsf{D}_{x^\prime,y^\prime}$ are timelike separated, then $\mathsf{K}_{x,y}\cap\mathsf{K}_{x^\prime,y^\prime}=\emptyset$ and $\supp \psi\cap\psi'=\emptyset$ because of the Huygens {principle}.
\end{remark}

\subsection{Modular Hamiltonian for deformed cones}\label{Subec: Modular Hamiltonian}
In this short section, we investigate the structure of the modular group for a distinguished class of deformed cones in $M$ defined as follows.

{
\begin{definition}\label{Def: Deformed Cone}
Let $C\in\cC^0(\mathbb{S}^2)$ and let $\mathsf{S}_C$ be the positive half-strip on $\Im^+$. The deformed cone associated to $C$ is defined as the causal completion of $\mathsf{S}_C$, i.e., $\mathsf{V}_C:= \mathsf{S}_C''$.
\end{definition}}

We {now investigate the relation between the} real subspace
$$N(\mathsf{V}_C) = \overline{\{\psi=\Upsilon_M([f])\in\cC_0^\infty(\Im^+)\;\textrm{with}\;\textrm{supp}(f)\subset \mathsf{V}_C\}}^{\|\cdot\|_\cH}\,$$
and  the standard subspaces $H(\mathsf{S}_C)$ defined as per Equation \eqref{def:H}. {To this end, we need a couple of technical lemmas. The first one gives an alternative description of deformed cones.

\begin{lemma}
The deformed cone $\mathsf{V}_C\subset M$ associated to $C\in\cC^0(\mathbb{S}^2)$  is the set of points $x \in M$ such that $\overline{\mathsf{K}_x} \subset \mathsf{S}_C$, where $\mathsf{K}_x$ is defined as per Equation \eqref{Eq: Kx}.
\end{lemma}

\begin{proof}
Let $\widetilde{\mathsf{V}}_C := \{ x \in M\,:\, \overline{\mathsf{K}_x} \subset \mathsf{S}_C\}$. We have to show that $\widetilde{\mathsf{V}}_C = \mathsf{S}_C''$. 
To begin with, one verifies, using Lemma~\ref{lem:KIm}, that  $x \in \widetilde{\mathsf{V}}_C$ if and only if $C_x(\theta, \varphi) > C(\theta, \varphi)$ for all $(\theta,\varphi) \in [0,\pi]\times[0,2\pi)$, so the continuity of $C$ implies that $\widetilde{\mathsf{V}}_C$ is open. 

Consider now $x \in \widetilde{\mathsf{V}}_C$. If $x \not \in \mathsf{S}''_C$, thanks to the previous observation we may assume, at the cost of replacing $x$ with some sufficiently near point, that $x$ can be connected with some $y \in \mathsf{S}_C'$ by a causal curve $\gamma$. Then an inextensible causal curve $\widetilde \gamma$ extending $\gamma$ will connect $y$ with  some point $z \in \overline{\mathsf{K}_x} \subset \mathsf{S}_C$, which is a contradiction, showing that $\widetilde{\mathsf{V}}_C \subset \mathsf{S}_C''$.

Conversely, consider $x \in \mathsf{S}''_C$ and assume by contradiction that there exists $y= (T_y,R_y,\bn_y) \in \overline{\mathsf{K}_x}$ such that $y \not \in \mathsf{S}_C$, that is, $u_y=\tan \left(\frac{T_y-R_y}2\right) \leq C(\bn_y)$. Since $\scri$ is a null hypersurface, then $y$ is spacelike or lightlike  separated from any other point of $\mathsf{S}_C$. For any point $z(\bn)=(T(\bn),R(\bn),\bn)\in \partial_{\scri} \mathsf{S}_C$ (so that $\tan(\frac{T(\bn)-R(\bn)}2) = C(\bn)$) there exists $z'(\bn)= (T'(\bn), R_y, \bn_y) \in \partial_{\widetilde M}J^+(z(\bn))$ such that $T'(\bn)>T_y$.  Moreover for every $T \in (T_y,T'(\bn))$ the point $(T, R_y,\bn_y)$ is spacelike to $z(\bn)$. Yet, since $T'(\bn)$ depends continuously on $\bn$ by the continuity of $C$, we can find $T$ such that $T_y < T < \min_{\bn\in\mathbb S^2} T'(\bn)$ and the point $z:=(T, R_y,\bn_y)$ belongs to $\mathsf{S}_C'$. At the same time, $z$ lies in the causal future of $y \in \overline{\mathsf{K}}_C$, and so also of $x \in \mathsf{S}_C''$, which is a contradiction.
\end{proof}

\noindent As an immediate consequence, if $x \in \mathsf{V}_C$, then $\mathsf{V}^+_x \subset \mathsf{V}_C$.
 }
\begin{lemma}
The collection $\{\mathsf{K}_x;\; x \in\mathsf{V}_C\}$ forms an open covering of $\mathsf{S}_C$.
\end{lemma}

\begin{proof}
Given $y \in\mathsf{S}_C$, we write $y = (u_y, \bn_y)$ with $u_y > C(\bn_y)$. It is sufficient to show that there exists $x =(t,\bx) \in M$ such that $C_x(\bn) = t-\bx \cdot \bn > C(\bn)$ for all $\bn \in \mathbb S^2$ and $C_x(\bn_y) < u_y$. Let $v = t+|\bx|$, $u = t-|\bx|$, $\boldsymbol{e} = \bx/|\bx|$ and choose $u$ and $\boldsymbol{e}$ such that $C(\bn_y) < u <u_y$, $\boldsymbol{e}\cdot \bn_y = 1$, so that $C_x(\bn_y) = u < u_y$. {Then, the condition $C_x(\bn)  > C(\bn)$, is equivalent to choose
\[
v > \frac{2 C(\bn)-u(1+\boldsymbol{e}\cdot \bn)}{1-\boldsymbol{e} \cdot \bn}, \quad \forall \bn \in \mathbb{S}^2
\]
which is possible since the right hand side is bounded from above on $\mathbb{S}^2$. We stress that $\boldsymbol{e} \cdot \bn=1$ if $\bn= \bn_y$ and $2 C(\bn)-u(1+\boldsymbol{e}\cdot \bn)<0$ for  $\bn$ close to $ \bn_y$.}
\end{proof}

\begin{proposition}\label{prop:VS} 
Under the same assumptions of Definition \ref{Def: Deformed Cone}, it holds that $N(\mathsf{V}_C)=H(\mathsf{S}_C)$.
\end{proposition}

\begin{proof}
The inclusion $N(\mathsf{V}_C) \subset H(\mathsf{S}_C)$ stems from the definition of $N(\mathsf{V}_C)$. To prove the converse inclusion, consider $\psi \in \mathcal{C}_0^\infty(\scri)$ such that $\supp \psi \subset S_C$. By the above lemma, we can find a finite family $x_1, \dots, x_n \in \mathsf{V}_C$ such that the collection of subsets $\{\mathsf{K}_{x_j}\}_{j=1,\dots, n}$ covers $\supp \psi$. Take $\varphi_j \in \mathcal{C}_0^\infty(\scri)$, $j=1, \dots n$, a partition of unity on $\supp \psi$ such that $\supp \varphi_j \subset \mathsf{K}_{x_j}$, $j=1,\dots, n$. Then, again by Proposition~\ref{prop:Vx}, $\varphi_j \psi \in N(\mathsf{V}_{x_j}^+)$, and therefore $\psi = \sum_{j=1}^n \varphi_j \psi \in N(\mathsf{V}_C)$.
\end{proof}

\noindent The following theorem is an immediate consequence of {Corollary~\ref{cor:stripdis},} Equation \eqref{eq:logint}, Remark~\ref{rem:propertiescurrent} and Theorem~\ref{theorem-1-standard}.
 
\begin{theorem}\label{main-thm}
Let $C\in\cC^0(\mathbb{S}^{2})$. Then, given $\mathsf{S}_C$ as per Equation \eqref{def:H}, the generator of the one particle modular group decomposes as follows
\begin{align*}
\log(\Delta_{N(\mathsf{V}_C)})=\log(\Delta_{H(\mathsf{S}_C)})&=\int^{\oplus}_{\mathbb{S}^{2}} \log(\Delta_{H^{(1)}(C(\theta,\varphi),+\infty)})d\mathbb{S}^2\\
&=\int^\oplus_{\bS^2}\left( \log(\Delta_{H^{(1)}(\RR_+)})+ 2\pi C(\theta,\varphi) P \right)\,d\mathbb{S}^2,
    \end{align*}
where $P$ is the generator of the translations {$s \mapsto U^{(1)}(\tau_{\RR_+}(s))$}. \\ {In particular for any $C_1,C_2\in \cC^0(\bS^2)$, such that $C_1<C_2$, then $N(\sV_{C_2})\subset N(\sV_{C_1})$ is an half-sided modular inclusion.}
\end{theorem}

\subsection{Relative entropy, ANEC and QNEC}\label{Subsec: QNEC}

We have recalled that Equation \eqref{eq:U1ent} establishes an expression for the relative entropy between two coherent states of the von Neumann algebra of an half-line of a single $U(1)$-current. We have shown that the  $U(1)$-current is the building block of the one particle theory living on $\scri$. We will follow an argument analogous to that in \cite{MTW} to derive an explicit formula for the relative entropy between coherent states of the deformed cone algebras {$\cA(\mathsf{V}_C)$ }of $M\subset\widetilde M$. It will be expressed in terms of the entropy of a one-particle vector state on $\scri$.

Let $C\in\cC^0(\mathbb S^2)$ identify the regions $\mathsf{S}_C$ and $\mathsf{V}_C$, in $\scri$ and $M$, respectively, see Equation \eqref{def:H} and Definition \ref{Def: Deformed Cone}. Given $f\in \cC_0^\infty(M)$ such that $\supp(f)\subset \mathsf{V}_C$, we can write the {relative entropy between the coherent state $\omega_{\Upsilon_M([f])} = \omega_{\scri} \circ \mathrm{Ad} \,W(\Upsilon_M([f]))$ and the vacuum $\omega_{\scri}$ on the algebra $R(H(\mathsf{S}_C))$}, as
{\begin{align}
S_{R(H(\mathsf{S}_C))}(\omega_{\Upsilon_M([f])}\|\omega_{\scri})
&=\int_{\mathbb S^2}\left(S_{R(H^{(1)}(C(\theta,\varphi),+\infty))}(\omega_{\Upsilon_M([f])(\cdot,\theta,\varphi)} \| \omega)\right)d\mathbb S^2\\
&= \int_{\bS^2} S_{H^{(1)}(C(\theta,\varphi),+\infty)}(\Upsilon_M([f])(\cdot,\theta,\varphi))\,d\bS^2 \\
&=\pi\int_{\mathbb S^2}\int_{C(\theta,\varphi)}^\infty (u-C(\theta,\varphi))(\partial_u\Upsilon_M([f]) (u,\theta,\varphi))^2 du \,d\mathbb S^2,
\end{align}
where the first equality follows from Equation \eqref{eq:HSC}, Lemma \ref{lem:decomposableentropy} and Theorem \ref{main-thm}, the second one from Equation~\eqref{eq:Ent2}, and the last one from Equation~\eqref{eq:U1ent}.}

Moreover, since $N(\mathsf{V}_C)=H(\mathsf{S}_C)$ on account of Proposition~\ref{prop:VS}, then the associated von Neumann algebras, defined as per Equations~{\eqref{eq:AO}, \eqref{eq:AOR}}, \eqref{Eq: VN algebras}, abide by ${\cA(\mathsf{V}_C)=}R(N(\mathsf{V}_C))=R(H(\mathsf{S}_C))$. Given $f\in \mathcal C_0^\infty(M)$ we have {then for the relative entropy between the coherent state $\omega_f = \omega_M \circ\operatorname{Ad}W([f])$ and the vacuum $\omega_M$ on $\cA(\mathsf{V}_C)$}
\begin{equation}\label{eq:entVc}\begin{split}
{S_{\cA(\mathsf{V}_C)}(\omega_f\|\omega_M)}&{=S_{R(H(\mathsf{S}_C))}(\omega_{\Upsilon_M([f])}\|\omega_{\scri})}\\
&=\pi\int_{\mathbb S^2}\int_{C(\theta,\varphi)}^\infty (u-C(\theta,\varphi))(\partial_u\Upsilon_M([f]) (u,\theta,\varphi))^2 du \,d\mathbb S^2.
\end{split}\end{equation} 
Note that this formula depends explicitly on $\Upsilon_M$ and, therefore, on $\widetilde{\mathcal{G}}$, the causal propagator of $\Box_{\tg}-\frac{\widetilde{R}}{6}$. 

Given $A:\bS^2\rightarrow\RR$ a positive continuous function we are interested in studying the convexity of the map
$$\RR\ni t\mapsto S_{{\cA(\mathsf{V}_{C+tA})}}(\rho\|\sigma)$$ 
where $\rho$ and $\sigma$ are two {coherent} states on {$\cA(\mathsf{V}_{C+tA})$} while $C\in\cC^0(\bS^{2})$. Observe that this quantity is strongly related to the quantum null energy condition (QNEC), see \cite[Sect.~5.3]{MTW} for a short survey as well as \cite{Hollands:2025glm}. In our setting, since we considered coherent states, we can check the convexity by studying the second derivative with respect to the deformation parameter $t$. 

\begin{theorem}\label{Thm: QNEC} Given Equation \eqref{eq:entVc} and $A:\bS^2\rightarrow\RR$ a positive continuous function,  it holds that
$$\frac{d^2}{dt^2}S_{{\cA(\mathsf{V}_{C+tA})}}(\omega_f\|\omega_M)\geq0.$$
\end{theorem}

\begin{proof} Starting from Equation \eqref{eq:entVc}, we observe that, for fixed $(\theta,\varphi)\in \mathbb S^2$, the maps 
\begin{gather*}
{\RR^2 \ni (u,t)} \mapsto (u-C(\theta,\varphi)-tA(\theta,\varphi))(\partial_u\Upsilon_M([f]) (u,\theta,\varphi))^2\\
{\RR^2 \ni (u,t) \mapsto \frac{\partial}{\partial t}\left[(u-C(\theta,\varphi)-tA(\theta,\varphi))(\partial_u\Upsilon_M([f]) (u,\theta,\varphi))^2\right] = -A(\theta,\varphi)(\partial_u\Upsilon_M([f]) (u,\theta,\varphi))^2}
\end{gather*}
{ are} continuous and {vanish for $u$ outside some compact interval since $f$ has compact support in $M$. Then, thanks to the Leibniz integral rule} it holds that
\begin{align*}
\frac{d}{dt} \int_{C(\theta,\varphi)+tA(\theta,\varphi)}^\infty(u-C(\theta,\varphi)&-tA(\theta,\varphi))(\partial_u \Upsilon_M([f]) (u,\theta,\varphi))^2du=&\notag\\ 
= -\int_{C(\theta,\varphi)+tA(\theta,\varphi)}^\infty &A(\theta,\varphi)(\partial_u\Upsilon_M([f]) (u,\theta,\varphi))^2 du.
\notag
\end{align*}
{Moreover, since $f \in \cC^\infty_0(M)$ and $C, A \in \cC^0(\bS^2)$, the right hand side of the above equation is clearly bounded by a constant uniformly in $(\theta,\varphi) \in \bS^2$, so that,} applying bounded convergence
it holds that
{
\[
\frac{d}{dt}S_{\cA(\mathsf{V}_{C+tA})}(\omega_f\|\omega_M) = -\pi \int_{\bS^2}\int_{C(\theta,\varphi)+tA(\theta,\varphi)}^\infty A(\theta,\varphi)(\partial_u\Upsilon_M([f]) (u,\theta,\varphi))^2 du \,d\bS^2.
\]
}
{Using again Leibniz integral rule and bounded convergence then gives}
\begin{align*}
\frac{d^2}{dt^2}S_{\cA(\mathsf{V}_{C+tA})}(\omega_f\|\omega_M)&=-\frac{d}{dt}\pi\int_{\mathbb S^2}\left(\int_{C(\theta,\varphi)+tA(\theta,\varphi)}^\infty A(\theta,\varphi)(\partial_u\Upsilon_M([f]) (u,\theta,\varphi))^2 du\right)\,d\mathbb S^2(\theta,\varphi)\\
&=\pi\int_{\mathbb S^2} A(\theta,\varphi)^2\big(\partial_u\Upsilon_M([f]) (C(\theta,\varphi)+tA(\theta,\varphi),\theta,\varphi)\big)^2 d\mathbb S^2(\theta,\varphi)\geq 0,
\end{align*}
which entails the sought conclusion.
\end{proof}

\begin{remark}
	We observe that, following \cite[Sec.\ 5.3]{MTW}, it is possible to deduce from the assumptions of Theorem \ref{Thm: QNEC} the following version of the {\em averaged null energy condition (ANEC)}:
\[
\langle W(\Upsilon_M([f]))\Omega,d\Gamma(H_A) W(\Upsilon_M([f]))\Omega\rangle	\geq 0,
\]
 where $d\Gamma(H_A)$ is the additive second quantization of $H_A=\int_{\mathbb S^2}^\oplus A(\theta,\varphi)P_{\theta,\varphi}d\mathbb S^2$, with $P_{\theta,\varphi}$ is the translation generator of the one-particle $U(1)$ current {and $A \in \cC^0(\bS^2)$ is a positive function}. Moreover, the saturation of the  {\em strong superadditivity} of the relative entropy for coherent states is obtained. More precisely, let $C_1,C_2\in\cC^0(\bS^2)$ and let $C_\lor(\theta,\varphi)=\min\{C_1(\theta,\varphi),C_2(\theta,\varphi)\}$ and $C_\land(\theta,\varphi)=\max\{C_1(\theta,\varphi),C_2(\theta,\varphi)\}$, respectively. Then, given two coherent states $\omega_1,\omega_2$ on ${\cA}(\mathsf{V}_{C_\lor})$, the  equality
	\begin{align*}
		S_{\cA(\mathsf{V}_{C_\lor})}(\omega_1\|\omega_2) + S_{\cA(\mathsf{V}_{C_\land})}(\omega_1\|\omega_2) 	= S_{\cA(\mathsf{V}_{C_1})}(\omega_1\|\omega_2) + S_{\cA(\mathsf{V}_{C_2})}(\omega_1\|\omega_2)
	\end{align*}
holds.
\end{remark}

{In view of the recalled connection with energy conditions, it is interesting to observe that the relative entropy~\eqref{eq:entVc} can be expressed in terms of the classical traceless stress-energy tensor associated to the solution $\Phi = \mathcal{G}(f)$ of the Klein-Gordon equation~\eqref{eq:wave} on $(M,g)$~\cite[Eq.\ (2.42)]{PT09}:
\[
T_{\mu \nu} = \nabla_\mu \Phi \nabla_\nu \Phi -\frac12 g_{\mu\nu} \nabla_\rho \Phi \nabla^\rho\Phi-\frac16 \left(R_{\mu\nu}-\frac12g_{\mu\nu}R\right)\Phi^2+\frac16[g_{\mu\nu}\Box -\nabla_\mu\nabla_\nu](\Phi^2).
\]

\begin{proposition}
Given $f\in \mathcal C_0^\infty(M)$, with $\Phi := \mathcal{G}(f)$ the associated solution of~\eqref{eq:wave}, it holds
\begin{multline*}
S_{\cA(\mathsf{V}_C)}(\omega_f\|\omega_M)\\
=\lim_{v \to +\infty}\pi\int_{\mathbb S^2}\int_{C(\theta,\varphi)}^\infty (u-C(\theta,\varphi))\left[\Omega^{-2} T_{uu} - \frac13 R_{uu}(\Omega^{-1}\Phi)^2 - \frac16 \partial_u^2(\Omega^{-2}\Phi^2)\right](u,v,\theta,\varphi)\,du \,d\mathbb S^2.
\end{multline*}
\end{proposition}

\begin{proof}
On account of the definition of $\Upsilon_M$ in Equation ~\eqref{Eq: Upsilon}, and in view of Equation ~\eqref{Eq:forseserveforseno},
\[
\Upsilon_M([f])(u, \theta,\varphi) = \lim_{v \to +\infty} (\Omega^{-1}\Phi)(u,v,\theta,\varphi).
\]
Proposition ~\ref{prop2} entails that $\Omega^{-1}\Phi$ extends to a smooth function on $\widetilde M$, whose restriction to $\scri$ has compact support. Hence we can interchange the limit with the derivative and the integral  in~\eqref{eq:entVc} to obtain
\[
S_{R(N(\mathsf{V}_C))}(\omega_M\circ\operatorname{Ad}f\|\omega_M)
=\lim_{v \to +\infty}\pi\int_{\mathbb S^2}\int_{C(\theta,\varphi)}^\infty (u-C(\theta,\varphi))(\partial_u (\Omega^{-1}\Phi))^2(u,v,\theta,\varphi)\,du \,d\mathbb S^2.
\]
Using that for $g= \chi^2 \eta$ one computes $R_{uu} = \chi^{-2}[4 (\partial_u \chi)^2-2\chi \partial_u^2 \chi]$ and the sought result is obtained by direct inspection.
\end{proof}

Analogous formulae connecting relative entropy of suitable states to the stress-energy tensor appear also in~\cite{CLR20, LM23, LPM25}, due to the connection between modular flow and Lorentz or conformal symmetries. It is not difficult to see that the same holds true for the relative entropy of null cuts computed in~\cite{MTW}. Moreover, in~\cite[Eq.\ (3.26)]{CTT17} the authors propose an expression for half of the modular Hamiltonian of a null cut on the light cone of Minkwoski space in terms of the quantum stress energy tensor of a conformal QFT, which is similar to the one in the proposition above.

{As a last comment we remark that the formula for the first quantization modular Hamiltonian given in Theorem \eqref{main-thm} has a second quantization counterpart in terms of continuous tensor product structure, (cf.~\cite{AW66, Napiorkowski71}) which is not easy to compare with formulas in \cite{CTT17}. On the other hand given $C_1,C_2\in \cC^0(\bS^2)$ such that $C_1<C_2$, by Theorem \ref{main-thm},  we have the second quantization von Neumann algebras  half-sided modular inclusion  $\cA(\sV_{C_2})\subset\cA(\sV_{C_1})$ with respect to the vacuum vector. This is the main ingredient to prove  the  QNEC in a model indepenent setting following \cite{CF18,Hollands:2025glm}. 
}
\section{Outlook}

In this paper we have considered a four-dimensional spacetime $(M,g)\equiv (\bR^4,\chi^2\eta)$ which is conformal to Minkowski spacetime $\bM$, so that its Penrose compactification allows to realize it as an embedded submanifold of the Einstein static Universe. Herein it possesses a non empty boundary and a particularly relevant component is $\Im^+$, future null infinity. On top of $M$ we have considered a massless, conformally coupled scalar field and its associated Weyl algebra $\mathcal{W}(M)$. Using a bulk-to-boundary correspondence, first discussed in \cite{Dappiaggi:2005ci} and here subordinated to the Huygens principle, we have realized $\mathcal{W}(M)$ as a C$^*$-subalgebra of $\mathcal{W}(\Im^+)$ which is the Weyl algebra associated to a symplectic space of kinematic configuration on $\Im^+$, see Theorem \ref{theo1}. The net advantage of this procedure is the possibility of associating a unique quasi-free state $\omega_{\Im^+}$ for $\mathcal{W}(\Im^+)$ which is invariant under the action of the group of asymptotic symmetries of $(M,g)$. In addition such state individuates a counterpart on $\mathcal{W}(M)$ which is quasi-free, Hadamard and invariant under the action of all bulk isometries, if present.

Using the GNS representation we can associate to $(\mathcal{W}(\Im^+),\omega_{\Im^+})$ a von Neumann algebra defined on the Bosonic Fock space built out of the one-particle Hilbert space $L^2(\RR_+\times\bS^2; EdEd\bS^2)$, see Equation \eqref{eq:hilbsp}. In addition, the Huygens principle entails that, if we consider $\mathsf{V}^+_x$, $x\in M$, a future light cone as per Definition \ref{def:regions} as well as $\mathcal{W}(\mathsf{V}^+_x)=\mathcal{W}(M)|_{\mathsf{V}^+_x}$, its counterpart at the boundary is the Weyl subalgebra generated by suitable functions localized in $\mathsf{K}_x$, a positive half strip on $\Im^+$. To each such cone, we associate a standard subspace of the boundary one-particle Hilbert space, see Equation \ref{def:N} and in Proposition \ref{prop:Vx} we have proven that this coincides with the standard subspace associated naturally to $\mathsf{K}_x$. We can extend such correspondence replacing $\mathsf{K}_x$ and $\mathsf{V}^+_x$ with deformed counterparts, respectively $\mathsf{S}_C$, see Equation \eqref{Eq: SC} and $\mathsf{V}_C$, see Definition \ref{Def: Deformed Cone}

Since the one particle Hilbert space at the boundary decomposes as a direct integral on the sphere of $U(1)$-currents defined on the half line, we prove in Theorem \ref{main-thm} that also the generator of the modular group associated to the standard subspace of $\mathsf{V}_C$ decomposes as a suitable direct integral. This result allows us to study the relative entropy between coherent states of the algebras associated to the deformed cones $\mathsf{V}_C$ establishing the quantum null energy condition in Theorem \ref{Thm: QNEC}.

This work can be seen as a starting point for several future investigations. First and foremost the bulk-to-boundary correspondence which lies at the heart of our analysis is valid on a generic globally hyperbolic and asymptotically flat spacetime, in the sense of \cite{Friedrich} as well on Schwarzschild spacetime \cite{Dappiaggi:2009fx}, where it has been used to give a rigorous definition of the Unruh state. Yet the Huygens principle does not hold true in general and, therefore, it is not straightforward how to extend to these cases our results in Section \ref{sect:4}. Staying instead in the context of spacetimes conformal to Minkowski, which includes the FLRW ones with flat spatial sections, one could envisage cosmological applications of our findings worth of further investigations. At last, we have focused our attention to a massless and conformally coupled scalar field, but it would be interesting to investigate other models, for example free electromagnetism written in terms of the Faraday tensor as well as Dirac fields, although, in this case the Huygens principle holds true only in odd spacetime dimensions.

\section*{Acknowledgments}
{C.D. and V.M. thank K.-H. Rehren for suggesting the subject of this work at the 46th LQP Workshop in Erlangen.} C.D. and A.R. are grateful to the Department of Mathematics of the University of Rome Tor Vergata for the kind hospitality during the realization of part of this work, V.M and A.R. for that of the Department of Physics of the University of Pavia. {A.R. acknowledges funding from the European Research Council (ERC) and the Swiss State Secretariat for Education, Research and Innovation (SERI) through the consolidator grant ProbQuant. A.R. is also grateful to the INFN, Sezione di Pavia for the support during the visit in Pavia.} The work of C.D. is partly supported by the GNFM (Indam). {The work of V.M. and G.M. is partly supported by INdAM-GNAMPA, University of Rome Tor Vergata funding \emph{OANGQS} CUP E83C25000580005, and the MIUR Excellence Department Project \emph{MatMod@TOV} awarded to the Department of Mathematics, University of Rome Tor Vergata, CUP E83C23000330006.}

\appendix

\section{Direct integrals of Hilbert spaces}\label{app:directintegral}

Since in our construction we also exploit the properties of direct integrals of standard subspaces, in this appendix we recollect some basic definitions and results, pointing an interested reader to \cite[Appendix B]{MT19} for additional details.

Let $X$ be a $\sigma$-compact, locally compact space endowed with $\nu$, the completion of a Borel measure thereon. In addition we consider a separable {complex (resp.\ real)} Hilbert space $\K$ together with $\{\K_\l  \}$ a family of separable  {complex (resp.\ real)} Hilbert spaces indexed by $\l\in X$. We also require that, to each $\xi\in \K$ there corresponds {linearly} a function $\l\mapsto \xi(\l) \in \K_\l$ such that 
\begin{enumerate}
	\item $\l\mapsto \langle\xi(\l),\eta(\l)\rangle$ is $\nu$-integrable and $\langle{\xi,\eta}\rangle=\int_X \langle\xi(\l),\eta(\l)\rangle d\nu(\l)$ for all $\xi, \eta \in \K$;
	\item if $\phi_\l \in \K_\l$ for all $\l$ and $\l\mapsto \langle{\phi_\l,\eta(\l)}\rangle$ is integrable for all $\eta\in \K$, there exists $\phi\in \K$ such that $\phi(\l)=\phi_\l$ for almost every $\l$.
\end{enumerate}
Under these assumptions $\K$ is the \textit{direct integral} of {the measurable field of Hilbert spaces} $\{ \K_\l\}$ over $(X,\nu)$, and we write
\begin{align*}
	\K=\int_X^{\oplus_{(\mathbb{R})}} \K_\l d\nu(\l),
\end{align*}
where we employ the subscript $\mathbb{R}$ when we consider direct integrals of real spaces. {The direct integral is uniquely determined, up to unitary equivalence, by $\{\K_\lambda\}$ and $\nu$. A $\xi \in \K$ corresponding as above to the function $\lambda \mapsto \xi(\lambda) \in \K_\lambda$, we will be written as $\xi = \int_X^\oplus \xi(\lambda)\,d\nu(\lambda)$.}

As a particular case, let $\mathcal{K}= \int_X^\oplus \mathcal{K}_0 d\nu$ be a direct integral Hilbert space over the constant field $\mathcal{K}_0$, where $\K$  is a separable Hilbert space. Then it holds that \cite[Proposition II.1.8.11, Corollary p. 175]{Dixmier81}
\begin{align}\label{integral=tensor-prod}
	\int_X^\oplus \mathcal{K}_0 d\nu(\l) \simeq L^2(X,\nu)\otimes \mathcal{K}_0.
\end{align}

\begin{lemma}\label{lem:propertiesstandard}\cite[Lemma B.3]{MT19}
	Given a separable Hilbert space $\cK$ which decomposes as $\cK=\int_X^\oplus \cK_\lambda d\nu(\lambda)$ then
	\begin{itemize}
		\item[(a)]
		given  $H=\int_X^{\oplus_\RR} H_\lambda d\nu(\lambda)\subset \cK$ such that  $H_\lambda\subset \cK_\lambda$ for all $\lambda \in X$,
		then $H'=\int_X^{\oplus_\RR} H'_\lambda d\nu(\lambda)$, $H^\prime$ being the symplectic complement;
		
		\item[(b)] let $\{H_k\}_{k\in\NN}$ be a countable family of real subspaces of $\cK$ such that  $H_k=\int_X^{\oplus_\RR} (H_k)_\lambda d\nu(\lambda)$
		and  $(H_k)_\lambda\subset \cK_\lambda$ {is a real subspace} for all $\lambda\in X$, then  $\bigcap_{k\in\NN} H_k=\int_X^{\oplus_\RR} \bigcap_{k\in\NN} (H_k)_\lambda d\nu(\lambda)$;
		
		\item[(c)]  Let $\{H_k\}_{k\in\NN}$ be a countable family of real subspaces of $\cK$ such that  $H_k=\int_X^{\oplus_\RR} (H_k)_\lambda d\nu(\lambda)$ and  $(H_k)_\lambda\subset \cK_\lambda$ {is a real subspace} for all $\lambda\in X$, then  $\overline{\Span_{k\in\NN} H_k}=\int_X^{\oplus_\RR} \overline{\Span_{k\in\NN} (H_k) }_\lambda d\nu(\lambda)$.
	\end{itemize}
\end{lemma}

A (possibily unbounded) operator $T$ on $\K = \int^{\oplus}_X \K_\lambda d\nu(\lambda)$ is said to be \textit{decomposable} if {there exist} (possibly unbounded) operators $T_\l$ on $\K_\lambda$ {, $\lambda \in X$, such that} for each $\xi \in D(T)$, one has $\xi(\lambda) \in D(T_\lambda)$ and $T_\l \xi(\l)=(T\xi)(\l)$ for $\nu$-almost every $\l$. {This is denoted by $T = \int_X^\oplus T_\l d\nu(\l)$.} In addition, if there exists $f\in L^\infty(X,\nu)$ such that $T_\l =f(\l)1$, we say that $T$ is \textit{diagonalizable}.

Let $G$ be a locally compact group and let $\pi$ be a strongly continuous unitary representation of $G$ on $\cK=\int^\oplus_X \cK_\lambda d\nu(\lambda)$. Suppose that, for each $g\in G$, $\pi(g)=\int^\oplus \pi_\lambda (g) d\mu(\lambda)$, then we say that $\pi$ is the \textit{direct integral} of $\pi_\lambda $, $\lambda\in X$: 
\begin{align*}
	\pi=\int^\oplus_X \pi_\lambda d\nu(\lambda).
\end{align*}
Equivalently, $\pi$ is a direct integral if each $\pi(g), \ g\in G$, is decomposable. Furthermore, if $H=\int_X^{\oplus_\RR} H_\lambda d\nu(\lambda)$ {is a standard subspace of $\K = \int_X^\oplus \K_\lambda d\nu(\lambda)$, with $H_\lambda$ standard in $\K_\lambda$, $\lambda \in X$}, we can apply these constructions to the modular data introduced in Section \ref{Sec: Standard Subspaces}:
$$\Delta^{it}_H=\int_X^\oplus\Delta^{it}_{H_\lambda}d\nu(\lambda),\quad J_H=\int_X^\oplus J_{H_\lambda}d\nu(\lambda),$$
where $( J_{H_\lambda},\Delta^{it}_{H_\lambda})$ are respectively the modular conjugation and the modular group associated to $H_\lambda $ as a standard subspace of $\cH_\lambda$. 
Given a Borel function $f:\RR_+\rightarrow\RR$, it holds that
	$$f(\Delta_H)=\int_{X}^\oplus f(\Delta_{H_\lambda})d\nu(\lambda),$$  
which entails in particular
	\begin{equation}\label{eq:logint}
		\log(\Delta_H)=\int_{X}^\oplus \log(\Delta_{H_\lambda})d\nu(\lambda),
		\end{equation}  
	see \cite[Thm.~A.1]{MTW}.

{In this context, it is worth observing the following decomposition result for the relative entropy of coherent states on a second quantization von Neumann algebra (cf.\ App.~\ref{app:secondquant}), whose proof is given in \cite[Lem.\ 5.1]{MTW}.

\begin{lemma}\label{lem:decomposableentropy}
Consider a standard subspace $H=\int_X^{\oplus_\RR}H_\lambda d\nu(\lambda)\subset \int_X^{\oplus} \K_\lambda d\nu(\lambda)=\K$, where $H_\lambda$ is standard in $\K_\lambda$, $\lambda \in X$. Given $\psi=\int_X^\oplus \psi(\lambda)d\nu(\lambda) \in \K$, the relative entropy between the coherent state $\omega_\psi$ and the Fock vacuum $\omega$ on the second quantization von Neumann algebra $R(H)$ can be written as
 $$S_{R(H)}(\omega_\psi\|\omega)=\int_X S_{R(H_\lambda)}(\omega_{\psi(\lambda)} \|\omega(\lambda))d\nu(\lambda)\,,$$ 
 where $\omega(\lambda)$ denotes the vacuum on the Fock space associated to $\K_\lambda$.
\end{lemma}
}

\section{Second quantization}\label{app:secondquant}
In this short appendix, we consider a complex Hilbert space $\cH$ and we recall some basic facts concerning the construction of the Weyl algebra associated to it. Our starting point is the Bosonic Fock space is $\mathcal{F}(\cH)=\bigoplus_{n=1}^\infty\cH^{\otimes^n_s}\oplus \bC$ where $\otimes^n_s$ is the $n$-times symmetrized tensor product. On top of it, one can define the Weyl unitary operators
$$\cH\ni h\mapsto W(h)\in \cF(\cH),$$
satisfying
$$W(h)W(k)=e^{-\frac12\Im(h,k)}W(h+k),\quad \omega(W(h))=(\Omega,W(h)\Omega)=e^{-\frac14\|h\|^2},$$
where $\Omega=(1,0,\dots)\in\mathcal{F}(\cH)$ is referred to as vacuum vector. Given a real subspace $H\subset \cH$, we denote by 
$R(H)$ the von Neumann algebra generated by the Weyl operators as 
\begin{equation}\label{Eq: VN algebras}
	R(H)=\{W(h):h\in H\}''\subset\cB(\cF(\cH)).
\end{equation}

\noindent Note that, by the strong continuity of Weyl operators, 
\[
R(H) = R(\overline{H})\ .
\]
Moreover $\Omega$ is cyclic (resp.\ separating) for $R(H)$ if and only if $\overline{H}$ is cyclic (resp.\ separating). If $H$ is standard, see Section \ref{Sec: Standard Subspaces}, we denote by $S_{R(H),\Omega}$, $J_{R(H),\Omega}$, $\Delta_{R(H),\Omega}$ the modular operators associated with $(R(H),\Omega)$, and by
$\mathit{\Gamma}(T)$  {the operator on $\cF(\cH)$ obtained by the multiplicative second quantization} of a one-particle operator $T$ on $\cH$. {Then one has the following properties, whose proof can be found, e.g., in~\cite{Lo08}.}

\begin{proposition}\label{prop:secquant}
	Let $H$ and $\{H_a\}_{a\in I}$ be {standard} subspaces of $\cH$. It holds that
	\begin{itemize}
		\item[$(a)$] $S_{R(H),\Omega} = \mathit{\Gamma}(S_H)$, \ $J_{R(H),\Omega} = \mathit{\Gamma}(J_H)$, \ $\Delta_{R(H),\Omega}= \mathit{\Gamma}(\Delta_H)$, 
		\item[$(b)$] $R(H)' = R(H')$, where $H^\prime$ is the symplectic complement,
		\item[$(c)$] $R(\sum_a H_a) = \bigvee_a R(H_a)$,
		\item[$(d)$] $R(\cap_a H_a) = \bigcap_a R(H_a)$,
	\end{itemize}
	where $\bigvee$ denotes the generated von Neumann algebra.
\end{proposition}

\section{Penrose Compactification of Minkowski space}\label{app:a}
{
In this appendix, we review succinctly the Penrose compactification of a four-dimensional, globally hyperbolic, spacetime $(M,g)$ such that $g=\chi^2\eta$ with $\chi\in \cC^\infty(\mathbb{R}^4,(0,\infty))$, hence providing an explicit example of the construction of $\scri$ sketched in Section \ref{sec:geometry}. We endow $M$ with the standard Cartesian coordinates $(t,\bx)\equiv (t,x_1,x_2,x_3)$ so that the line element reads
$$ds^2=\chi^2\left(-dt^2+dx_1^2+dx_2^2+dx_3^2.\right),$$ 
where we omit for simplicity of the notation the dependence of $\chi$ on the underlying coordinates, that is $\chi\equiv\chi(t,x_1,x_2,x_3)$. Switching to the standard spherical coordinates, namely 
\begin{equation}
\left\{ \begin{aligned} 
  x_1 &= r\cos\theta\sin\varphi\\
  x_2 &= r\sin\theta\sin\varphi\\
  x_3 &= r\cos\varphi
  \end{aligned} \right.
\end{equation}
where $r^2=\|\bx\|^2$ while $(\theta,\varphi)\in \mathbb{S}^2$, we obtain 
$$ds^2=\chi^2\left(-dt^2 + dr^2 + r^2(d\theta^2 + \sin^2\theta d \varphi^2)\right)$$ 
\noindent We introduce the advanced and retarded null coordinates 
\begin{equation}\label{eq:nullcoord}
\left\{ \begin{aligned} 
  u &= t-r\\
  v &= t+r
  \end{aligned} \right.
\end{equation}

\noindent In the new coordinates $(u,v,\theta,\varphi)$ the line element becomes 
$$ds^2 = \chi^2\left(-dudv + \frac{1}{4}(v-u)^2(d\theta^2 + \sin^2\theta d\varphi^2)\right)\,.$$ 
We rescale these new coordinates introducing $U,V \in (-\frac{\pi}{2},\frac{\pi}{2})$ such that 
\begin{equation}\label{Eq: U and V}
\left\{ \begin{aligned} 
  u &= \tan(U)\\
  v &= \tan(V)
  \end{aligned} \right. .
\end{equation}

\noindent With this choice we have $du=\frac{1}{\cos^2(U)}dU$, and analogously $dv=\frac{1}{\cos^2(V)}dV$; therefore
$$ds^2=\frac{\chi^2}{\cos^2(V)\cos^2(U)}\left(-dUdV + \frac{1}{4}\sin^2(V-U)(d\theta^2 + \sin^2\theta d\varphi^2) \right)\,.$$

\noindent In order to choose a conformal compactification of $M$, we set $\Omega(U,V,\theta,\varphi):=\chi^{-1}\cos(U)\cos(V)$ so that the line element associated to $\widetilde{g}:=\Omega^2 g$ reads
$$d\widetilde{s}^2=-dUdV + \frac{1}{4}\sin^2(V-U)(d\theta^2 + \sin^2\theta d\varphi^2)\,,$$
or, equivalently, introducing \begin{equation}
\left\{ \begin{aligned} 
  T&= V+U\\
  R &= V-U
  \end{aligned} \right. ,
\end{equation}
\begin{equation}\label{eq:einsteinstatic}
d\widetilde{s}^2=-dT^2+dR^2+\sin^2 R(d\theta^2 + \sin^2\theta d\varphi^2)\,.
\end{equation}
This allows to identify $(M,g)$ as a (conformal) subset of the Einstein static Universe. Observe that $R>0$ since $V-U=\arctan(v)-\arctan(u)$ and $v-u=2r>0$.}\\

\begin{lemma}\label{lem:limitcomponents}
Let $x=(t,\bx) \in M$ with $\bx=(x_1,x_2,x_3)$. Consider a unit vector $\bn\in\mathbb{S}^2$ and let $\xi_{\lambda,\bn}=(\lambda,\lambda\bn)$, $\lambda\in\bR$. Denote with $(u,\theta,\varphi)$ the coordinates on $\Im^+$. Then $$\Im^+ \ni \lim\limits_{\lambda\to +\infty} x+\xi_{\lambda,\bn}=(t-\bx\cdot\bn, \bn)\,.$$ 
\end{lemma}

\begin{proof}
Denote $u_{x,\bn}(\lambda)$ the first coordinate of $x+\xi_{\lambda,\bn}$ in the null coordinates $(u,v)$, see Equation $\eqref{eq:nullcoord}$. Then, 
\begin{eqnarray}
\nonumber
u_{x,\bn}(\lambda)&:=&(t+\lambda) - \|\bx-\lambda\bn\|=(t+\lambda) - \sqrt{\sum_{i=1}^{3}(x_i + \lambda n_i)^2}\\
\nonumber
& = & (t+\lambda)-\sqrt{\|\bx\|^2 + \|\lambda\bn\|^2 + 2\lambda\,\bx\cdot \bn}\\
\nonumber
&=& \left((t+\lambda)^2-\|\bx\|^2 - \|\lambda\bn\|^2 - 2\lambda\,\bx\cdot \bn\right)\cdot \left((t+\lambda)+\sqrt{\|\bx\|^2 + \|\lambda\bn\|^2 + 2\lambda\,\bx\cdot \bn}\right)^{-1}\\
\nonumber
&=& \left(\lambda^2+2\lambda t + t^2 - \|\bx\|^2 - \lambda^2 - 2\lambda \bx\cdot \bn\right)\cdot \lambda^{-1}\cdot \left(\frac{t}{\lambda} +1 + \sqrt{\frac{\|\bx\|^2}{\lambda^2} + 1 + \frac{2}{\lambda}\bx\cdot \bn}\right)^{-1}\\
&=&\frac{2(t-\bx\cdot\bn) +\frac{t^2-\|\bx\|^2}{\lambda}}{\frac{t}{\lambda} +1 + \sqrt{\frac{\|\bx\|^2}{\lambda^2} + 1 + \frac{2}{\lambda}\bx\cdot \bn}}
\end{eqnarray}

\noindent Therefore 

\begin{equation}
\nonumber
\lim\limits_{\lambda\to+\infty}u_{x,\bn}(\lambda)=t - \bx\cdot\bn\,.
\end{equation}

\noindent Similar manipulations yield that also the angular components of $x+\xi_{\lambda,\bn}$ converge to those of $\bn$, completing the proof. 
\end{proof}

\end{document}